%% file: cold_start_rect_maxvol_icdm.tex
\documentclass[conference]{IEEEtran}

\usepackage{algorithm,algorithmic}
\usepackage{pgf}
\usepackage{multicol}

\newcommand{\argmin}{\operatornamewithlimits{argmin}}
\newcommand{\argmax}{\operatornamewithlimits{argmax}}

\usepackage[colorinlistoftodos,textwidth=4cm,shadow]{todonotes}

\usepackage[colorlinks,linkcolor=orange,urlcolor=blue,citecolor=purple]{hyperref}
\usepackage{algorithmic,algorithm}
\usepackage{ upgreek }
\usepackage{amsfonts}
\usepackage{amsmath}
\usepackage{graphicx} 
\usepackage{epsfig}
\usepackage{amsthm}
\usepackage{url}

\begin{document}

\title{Efficient Rectangular Maximal-Volume\\Algorithm for Rating Elicitation\\in Collaborative Filtering}

\author{\IEEEauthorblockN{Alexander Fonarev\IEEEauthorrefmark{1}\IEEEauthorrefmark{2}\IEEEauthorrefmark{3}, Alexander Mikhalev\IEEEauthorrefmark{4}, Pavel Serdyukov\IEEEauthorrefmark{2}, Gleb Gusev\IEEEauthorrefmark{2}, Ivan Oseledets\IEEEauthorrefmark{1}\IEEEauthorrefmark{5}}
\IEEEauthorblockA{\IEEEauthorrefmark{1}Skolkovo Institute of Science and Technology, Moscow, Russia\\
\IEEEauthorrefmark{2}Yandex, Moscow, Russia\\
\IEEEauthorrefmark{3}SBDA Group, Dublin, Ireland\\
\IEEEauthorrefmark{4}King Abdullah University of Science and Technology, Thuwal, Kingdom of Saudi Arabia\\
\IEEEauthorrefmark{5} Institute of Numerical Mathematics, Russian Academy of Sciences, Moscow, Russia\\
Email: newo@newo.su, aleksandr.mikhalev@kaust.edu.sa, \{pavser, gleb57\}@yandex-team.ru, i.oseledets@skoltech.ru}
}

\maketitle

\begin{abstract}

Cold start problem in Collaborative Filtering can be solved by asking new users to rate a small seed set of representative items or by asking representative users to rate a new item. The question is how to build a seed set that can give enough preference information for making good recommendations. One of the most successful approaches, called Representative Based Matrix Factorization, is based on Maxvol algorithm. Unfortunately, this approach has one important limitation --- a seed set of a particular size requires a rating matrix factorization of fixed rank that should coincide with that size. This is not necessarily optimal in the general case. In the current paper, we introduce a fast algorithm for an analytical generalization of this approach that we call Rectangular Maxvol. It allows the rank of factorization to be lower than the required size of the seed set. Moreover, the paper includes the theoretical analysis of the method's error, the complexity analysis of the existing methods and the comparison to the state-of-the-art approaches.
\end{abstract}

\IEEEpeerreviewmaketitle

\section{Introduction}

Collaborative Filtering (CF)~\cite{resnick1994grouplens} is one of the most widely used approaches to recommender systems. It is based on the analysis of users' previous activity (likes, watches, skips, etc.\ of items) and discovering hidden relations between users and items. Among CF methods, matrix factorization techniques~\cite{koren2008factorization, koren2009matrix} offer the most competitive performance \cite{cremonesi2010performance}. These models map users and items into a latent factor space which contains information about preferences of users w.r.t.\ items. 
Due to the fact that CF approaches use only user behavioural data for predictions, but not any domain-specific context of users/items, they cannot generate recommendations for new \textit{cold} users or \textit{cold} items which have no ratings so far.

A very common approach to solve this \textit{cold-start problem}~\cite{schein2002methods}, called \textit{rating elicitation}, is to explicitly ask cold users to rate a small representative \textit{seed set} of items or to ask a representative \textit{seed set} of users to rate a cold item \cite{golbandi2010bootstrapping, golbandi2011adaptive, liu2011wisdom}. One of the most successful approaches \cite{liu2011wisdom} to rating elicitation is based on the maximal-volume concept \cite{goreinov2001maximal}.
Its general intuition is that the most representative seed set should consist of the most representative and diverse latent vectors, i.e. they should have the largest length yet be as orthogonal as possible to each other. Formally, the degree to which these two requirements are met is measured by the volume of the parallelepiped spanned by these latent vectors. In matrix terms, the algorithm, called Maxvol \cite{goreinov2010find}, searches very efficiently for a submatrix of a factor matrix with the locally maximal determinant. Unfortunately, the determinant is defined only for square matrices, what means that a given fixed size of a seed set requires the same rank of the matrix factorization that may be not optimal. For example, the search for a sufficiently large seed set requires a relatively high rank of factorization, and hence a higher rank implies a larger number of the model parameters and a higher risk of overfitting, which, in turn, decreases the quality of recommendations.

To overcome the intrinsic ``squareness'' of the ordinary Maxvol, which is entirely based on the determinant, we use the notion of rectangular matrix volume, a generalization of the usual determinant. 
Searching a submatrix with high rectangular volume allows to use ranks of the factorization that are lower than the size of a seed set. However, the problem of searching for the globally optimal rectangular submatrix is NP-hard in the general case. In this paper, we propose a novel efficient algorithm, called Rectangular Maxvol, which generalizes original Maxvol.

It works in a greedy fashion and adds representative objects into a seed set one by one. This incremental update has low computational complexity that results in high algorithm efficiency. In this paper, we provide a detailed complexity analysis of the algorithm and its competitors and present a theoretical analysis of its error bounds. Moreover, as demonstrated by our experiments, the rectangular volume notion leads to a noticeable quality improvement of recommendations on popular recommender datasets. 

Let us briefly describe the organisation of the paper. Section~\ref{sec:background} describes the background on the existing methods for searching representatives that is required for further understanding. Sections \ref{sec:volume}, \ref{sec:algorithm} present our novel approach based on the notion of rectangular matrix volume and the fast algorithm to search for submatrices with submaximal volume. In Sections \ref{sec:rectmaxvol_complexity} and \ref{sec:bound}, we provide a theoretical analysis of the proposed method. Section \ref{sec:experiments} reports the results of our experiments conducted on several large-scale real world datasets. Section \ref{sec:related} overviews the existing literature related to CF, the cold start problem and the basic maximal-volume concept papers.

\section{Background and Framework}
\label{sec:background}

\subsection{Rating elicitation scheme}

The rating elicitation methods, such as \cite{golbandi2010bootstrapping, liu2011wisdom, anava2015budget, rashid2002getting},
 are based on the same common scheme, which is introduced in this section. Suppose we have a system that contains a history of users' ratings for items, where only a few items may be rated by a particular user. Denote the rating matrix by $R \in \mathbb{R}^{n \times m}$, where $n$ is the number of users and $m$ is the number of items, and the value of its entry $r_{ui}$ describes the feedback of user $u$ on item $i$. If the rating for pair $(u,i)$ is unknown, then $r_{ui}$ is set to~$0$.
Without loss of generality and due to the space limit, the following description of the methods is provided only for the user cold start problem. Without any modifications, these methods for the user cold start problem can be used to solve the item cold start problem after the transposition of matrix $R$.

Algorithm \ref{alg:elicitation} presents the general scheme of a rating elicitation method. Such procedures ask a cold user to rate a \textit{seed set} of representative items with indices $k \in \mathbb{N}^{L_0}$ for modeling his preference characteristics, where $L_0$, called \textit{budget}, is a parameter of the rating elicitation system.

\begin{algorithm}[h]
\caption{Rating elicitation for user cold start problem}
\label{alg:elicitation}
\begin{algorithmic}[1]
\REQUIRE Warm rating matrix $R \in \mathbb{R}^{n \times m}$, cold user, budget $L_0$
\ENSURE Predicted ratings of the cold user for all items
\STATE Compute indices $k \in \mathbb{N}^{L_0}$ of representative items that form a seed set
\STATE Elicit ratings $z' \in \mathbb{R}^{1\times L_0}$ of the cold user on items with indices $k$
\STATE Predict ratings of the cold user for all items $z \in\mathbb{R}^{1\times m}$ using $z'$
\RETURN $z$
\end{algorithmic}
\end{algorithm}

The performance of a rating elicitation procedure should be measured using a quality of predictions $z$. For this purpose, we use ranking measures (such as Precision@$k$), which are well suitable for CF task (see Section \ref{sec:experiments} for details).

The major contribution of this paper is a novel method of performing Step 1, described in Section \ref{sec:proposed_method}. It is based on PureSVD \cite{cremonesi2010performance} Collaborative Filtering technique, that is described in Section \ref{sec:puresvd}. In Section \ref{sec:predicting}, we discuss how to effectively perform Step 3 using the similar factorization based approach. And in Section \ref{sec:maximal}, we talk about the baseline method for seeking a seed set (Step 1), which is based on the maximal-volume concept.
\subsection{PureSVD}
\label{sec:puresvd}

Let us briefly describe the general idea of PureSVD, which is a very effective CF method in terms of ranking measures \cite{cremonesi2010performance} and therefore used as a basis of our rating elicitation approach. PureSVD provides a solution of the following optimization problem:
\begin{equation}
\label{eq:puresvd_optimization}
\Vert R-P^\top Q\Vert_F\to\min_{\substack{P\in\mathbb{R}^{f\times n}\\Q\in\mathbb{R}^{f\times m}}},
\end{equation}
where $||\cdot||_F$ is the Frobenius norm and $f$ is a parameter of PureSVD called rank. According to Eckart-Young theorem~\cite{golub2012matrix}, the optimal solution can be found by computing the truncated sparse Singular Value Decomposition of the sparse rating matrix $R$. 

This factorization can be interpreted as follows. Every user~$u$ has a low dimensional embedding $p_u \in \mathbb{R}^f$, a row in the matrix $P$, and every item has an embedding $q_i \in \mathbb{R}^f$, a column of the matrix $Q$. These embeddings are called \textit{latent vectors}~\cite{koren2008factorization}.
The PureSVD method provides an approximation~$\tilde r_{ui}$ of the unknown rating for a pair $(u,i)$, which is computed as the scalar product of the latent vectors: $$\tilde r_{ui} = p_u^\top  q_i.$$
Low-rank factors $P$ and $Q$ are used in the rating elicitation procedures that are described further.

\subsection{Predicting Ratings with a Seed Set}
\label{sec:predicting}

Let us assume that some algorithm has selected a seed set with $L_0$ representative items with indices $k \in \mathbb{N}^{L_0}$, and assume a cold user has been asked to rate only items~$k$, according to Steps 1-2 of the rating elicitation scheme described by Algorithm~\ref{alg:elicitation}. In this section, we explain how to perform Step 3, i.e. how to predict ratings $z$ for all items using only the ratings of the seed set.

As shown in \cite{anava2015budget}, the most accurate way to do it is to find a coefficient matrix $C\in\mathbb{R}^{L_0\times m}$ that allows to linearly approximate each item rating via ratings $z'$ of items from the seed set. Each column of $C$ contains the coefficients of the representation of an item rating via the ratings of the items from the seed set. Shortly, this approximation can be written in the following way:
\begin{equation*}
z\leftarrow z'C.
\end{equation*}
We highlight two different approaches to compute matrix C.

\subsubsection{Computing coefficients via the rating matrix} First approach is called Representative Based Matrix Factorization~(RBMF) \cite{liu2011wisdom}. It aims to solve the following optimization task:
\begin{equation}
\label{eq:matrix_decomp}
\Vert R-R(:,k) C\Vert_F \to \min_C.
\end{equation}
In our paper, we use the Matlab indexing notation\footnote{\url{http://www.mathworks.com/company/newsletters/articles/matrix-indexing-in-matlab.html}}: $R(:,k)$ is the matrix whose column $j$ coincides with the column $k_j$ of~$R$, where $k_j$ is the $j$th component of vector $k$. Note that $z'$ is not a part of $R(:,k)$, because there is still no information about a cold user ratings.
This optimization task corresponds to the following approximation:
\begin{equation}
\label{eq:r_coefs}
R\approx R(:,k)C.
\end{equation}
The solution of (\ref{eq:matrix_decomp}) is:
\begin{equation}
\label{eq:coef_matrix}
\begin{split}
C_R = (R(:,k)^\top  R&(:, k))^{-1}R(:,k)^\top R.
\end{split}
\end{equation}
Since $L_0 \ll n$, the matrix $R(:,k)$ is often well-conditioned. Therefore, the regularization term used in \cite{liu2011wisdom} is unnecessary and does not give a quality gain.

\begin{figure*}[t]
\begin{multicols}{2}
\hfill
\hspace{-4mm}
\epsfig{file=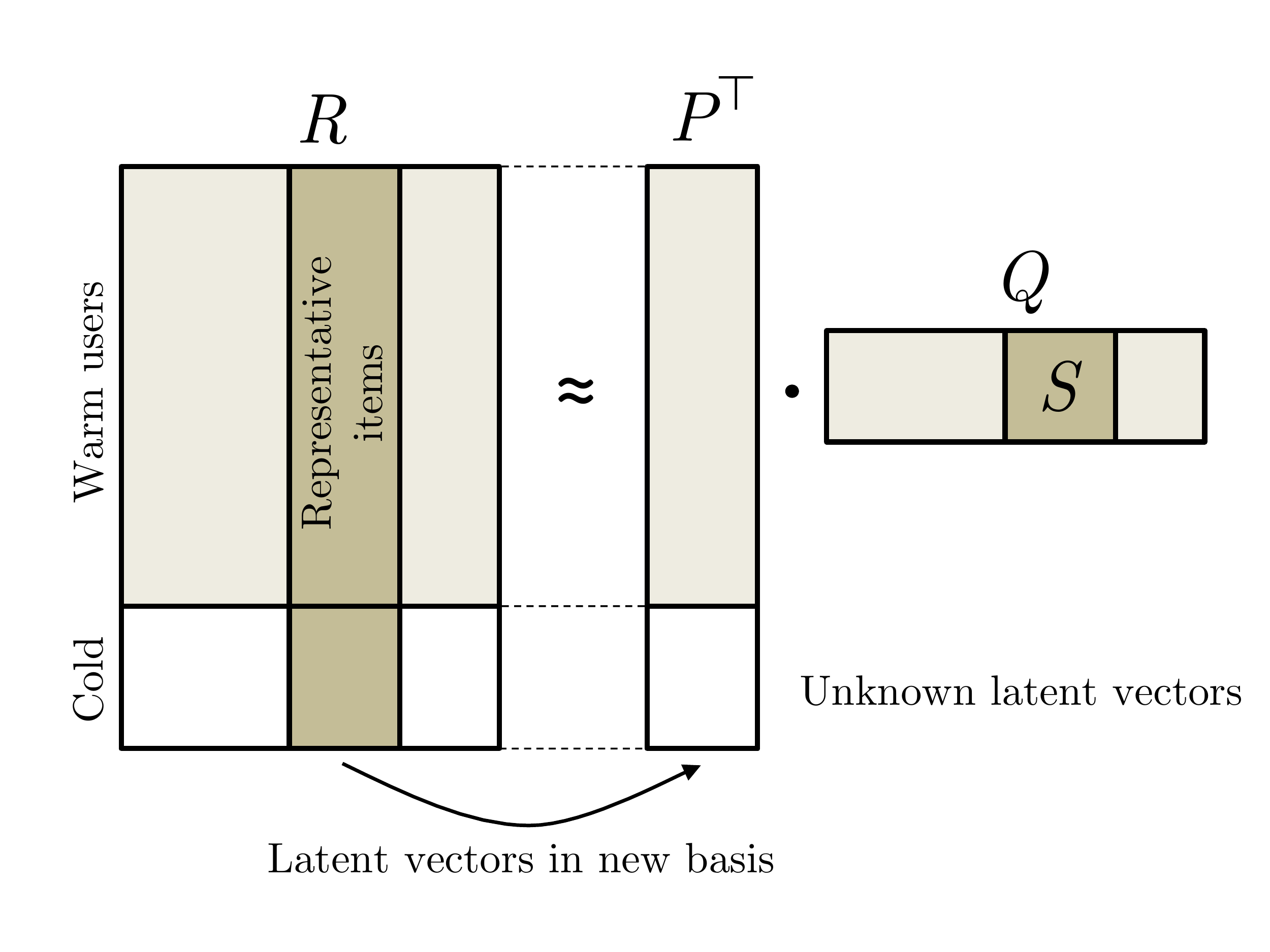, height=6cm}
\caption{Rating prediction using the seed set}
\label{fig:scheme}
\hfill
\epsfig{file=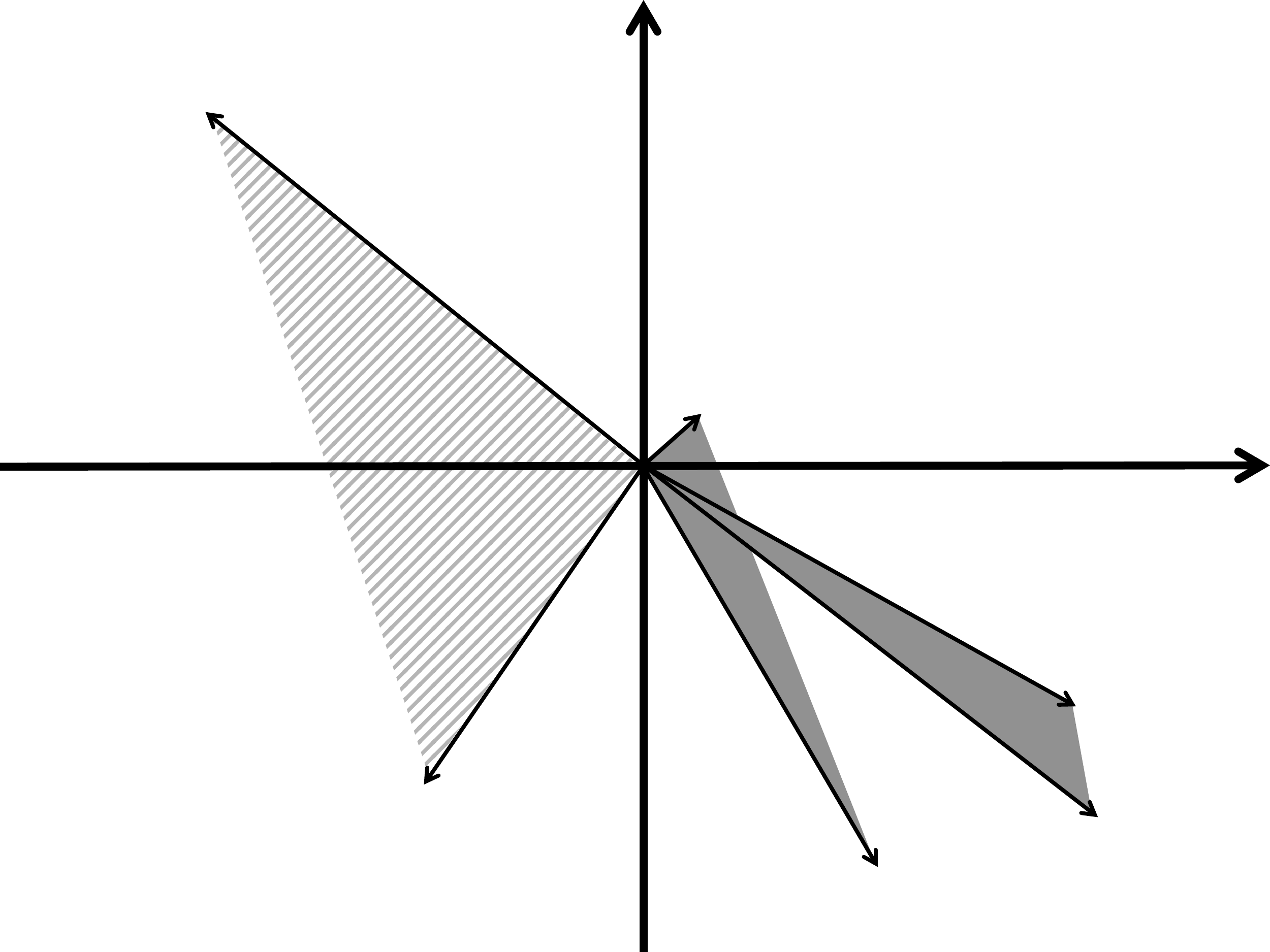, height=5.5cm}
\caption{An illustration of the intuition behind Maxvol for searching the seed set}
\label{fig:maxvol_demo}
\end{multicols}
\end{figure*}

\subsubsection{Computing coefficients via a low-rank factor} In this paper, we propose a more efficient second approach that considers the rank-$f$ factorization given by Equation (\ref{eq:puresvd_optimization}),~$f\leq L_0$. Let $Q(:,k) \in \mathbb{R}^{f \times L_0}$ be the matrix formed by $L_0$ columns of~$Q$ that correspond to the items of the seed set. Let us try to linearly recover all item latent vectors via the latent vectors from the seed set:
\begin{equation}
\label{eq:coef_factor_optimization}
\Vert Q - Q(:,k)C\Vert_F \to \min_C.
\end{equation}
It is a low-rank version of the problem given by (\ref{eq:matrix_decomp}) and, therefore, is computationally easier. Solution $C$ of this optimization problem can be also used for recovering all ratings using (\ref{eq:r_coefs}). 

Unlike (\ref{eq:matrix_decomp}), the optimization problem given by (\ref{eq:coef_factor_optimization}) does not have a unique solution $C$ in general case, because there are infinitely many ways to linearly represent an $f$-dimensional vector via more than $f$ other vectors. Therefore, we should find a solution of the underdetermined system of linear equations:
\begin{equation}
\label{eq:coef}
Q = SC,
\end{equation}
where we denote $S=Q(:, k)$. Since the seed set latent vectors surely contain some noise and coefficients in $C$ show how all item latent vectors depend on the seed set latent vectors, it is natural to find ``small'' $C$, because larger coefficients produce larger noise in predictions. We use the least-norm solution $C$ in our research, what is additionally theoretically motivated in Section \ref{sec:bound}. The least-norm solution of (\ref{eq:coef}) should be computed as follows:
\begin{equation}
\label{eq:coef_pseudo}
C=S^{\dagger}Q,
\end{equation}
where $S^{\dagger}=S^\top(SS^\top )^{-1}$ is the right pseudo-inverse of $S$.

Actually, such linear approach to rating recovering results in the following factorization model. 
Taking the latent vectors of the representative items $S$ as a new basis of the decomposition given by Equation (\ref{eq:puresvd_optimization}), we have
\begin{equation*}
\begin{split}
R \approx P^\top Q = P^\top SC = \left(P^\top Q(:,k)\right)C\approx\\
\approx R(:,k)C=F^\top C,\end{split}
\end{equation*}
where $F^\top =R(:,k)$. In this way, we approximate an unknown rating $r_{ui}$ by the corresponding entry of matrix $F^\top C$, where factor $F$ consists of the known ratings for the seed set items. This scheme is illustrated on Fig. \ref{fig:scheme}.

\subsection{Square Maxvol}
\label{sec:maximal}

This section introduces the general idea of the maximal-volume concept and Maxvol algorithm \cite{liu2011wisdom} for selecting a good seed set, what corresponds to Step 1 in the rating elicitation scheme (Algorithm \ref{alg:elicitation}).

Suppose we want to select $L_0$ representative items with indices $k \in \mathbb{N}^{L_0}$.
First of all, Maxvol algorithm requires to compute the rank-$L_0$ SVD factorization of $R$ given by Equation (\ref{eq:puresvd_optimization}). After this, searching for an item seed set is equivalent to searching for a  square submatrix $S =Q(:,k)\in \mathbb{R}^{L_0\times L_0}$ in the factor matrix $Q$. Note that every column of $S$ or $Q$  is a latent vector corresponding to an item from the seed set. 

An algorithm of seeking for a set of representative items may rely on the following intuitions. First, it should not select items, if they are not popular and thus cover preferences of only a small non-representative group of users. That means that the latent vectors from the seed set should have large norms. Second, the algorithm has to select diverse items that are relevant to different users with different tastes. This can be formalized as selecting latent vectors that are far from being collinear. The requirements can be met by searching for a subset of columns of $Q$ that maximizes the volume of the parallelepiped spanned by them. This intuition is demonstrated in Fig. \ref{fig:maxvol_demo},
which captures a two-dimensional latent space and three seed sets. The volume of each seed set is proportional to the area of the triangle built on the corresponding latent vectors. The dark grey triangles have small volumes (because they contain not diverse vectors or vectors with small length) and hence correspond to bad seed sets. Contrariwise, the light gray triangle has a large volume and represents a better seed set. 

Overall, we have the following optimization task:
\begin{equation}
\label{eq:optimize_det}
k \leftarrow \argmax_{k} \text{Vol}~S= \argmax_{k} |\det S|,\quad S=Q(:,k).
\end{equation}
The problem is NP-hard in the general case \cite{civril2007finding} and, therefore, suboptimal greedy procedures are usually applied. One of the most popular procedures is called Maxvol algorithm~\cite{goreinov2010find} and is based on searching for a \textit{dominant} submatrix $S\in\mathbb{R}^{L_0 \times L_0}$ of $Q$. The dominant property of $S$ means that all columns~$q_i\in\mathbb{R}^{L_0}$ of $Q$ can be represented via a linear combination of columns from $S$ with the coefficients not greater than 1 in modulus. Although, this property does not imply that $S$ has the maximal volume, it guarantees that $S$ is \textit{locally optimal}, what means that replacing any column of $S$ with a column of~$Q$, does not increase the volume \cite{goreinov2010find}.

At the initialization step, Maxvol takes $L_0$ linearly independent latent vectors that are the pivots from LU-decomposition~\cite{golub2012matrix} of matrix $Q$. Practice shows that this initialization usually provides a good initial approximation~$S$ to maximal volume matrix \cite{goreinov2010find}. After this, the algorithm iteratively swaps a ``bad'' latent vector inside the seed set with a~``good'' one out of it.
The procedure repeats until convergence. See~\cite{goreinov2010find} for more rigorous explanation of Maxvol algorithm.
In our paper, we also call this algorithm \textit{Square Maxvol}, because it seeks for a square submatrix (since determinant is defined only for square $S$). Furthermore, it is important to note that the original algorithm presented in \cite{goreinov2010find} has crucial speed optimizations for avoiding the expensive matrix multiplications and inversions, which are not presented in our paper due to the lack of space.

Let us analyse the complexity of Maxvol. The LU-decomposition with pivoting takes $O(mL_0^2)$ operations. The iterative updates take $O(\alpha mL_0)$ operations, where $\alpha$ is the number of iterations. Typically, $\alpha \le L_0$ iterations are needed. The overall complexity of Square Maxvol can be estimated as $O(mL_0^2)$. A more detailed complexity analysis of Square Maxvol is given in \cite{goreinov2010find}.

The obvious disadvantage of this approach to rating elicitation is the fixed size of the decomposition rank $f=L_0$, because the matrix determinant is defined only for square matrices. That makes it impossible to build a seed set with fixed size $L_0$ using an arbitrary rank of decomposition. However, as we further demonstrate in Section \ref{sec:experiments} with experiments, using our Rectangular Maxvol generalization with a decomposition of rank $f$ smaller than the size $L_0$ of the seed set could result in better accuracy of recommendations for cold users.

\section{Proposed Method}
\label{sec:proposed_method}

\subsection{Volume of Rectangular Matrices}
\label{sec:volume}

This section introduces a generalization of the maximal-volume concept to rectangular submatrices, which allows to overcome the intrinsic ``squareness'' of the ordinary maximal-volume concept, which is entirely based on the determinant of a square matrix.

Consider $S \in \mathbb{R}^{f \times L_0}$, $f \le L_0$. It is easy to see that the volume of a square matrix is equal to the product of its singular values. In the case of a rectangular matrix $S$, its volume \cite{rose1982linear} can be defined in a similar way:
$$\text{Rectvol}(S) := \prod_{s=1}^{L_0} \sigma_s = \sqrt{\det(SS^\top )}.$$
We call it \textit{rectangular volume}.
The simple intuition behind this definition is that it is the volume of the ellipsoid defined as the image of a unit sphere under the linear transformation defined by $S$:
$$\text{Rectvol}(S) = \text{Vol}~\{v \in \mathbb{R}^f : \exists c \in \mathbb{R}^L_0, \Vert c \Vert_2 \le 1 \,|\, v=Sc\}.$$
This can be verified using the singular value decomposition of~$S$ and the unitary invariance of the spectral norm. Moreover, in the case of a square matrix $S$, the rectangular volume is equal to the ordinary square volume:
$$\text{Rectvol}(S) = \sqrt{\det(SS^\top )}=|\det(S)|=\text{Vol}(S).$$
Note that, if $f > L_0$, then $\det SS^\top =0$.

Overall, searching for a seed set transforms to the following optimization task that is a generalization of Problem~(\ref{eq:optimize_det}): $k \leftarrow \argmax_{k}~\text{Rectvol}(S)$, where $S=Q(:,k)$. It is important to note that this maximization problem does not depend on the basis of the latent vectors from $S$.

The simplest method to find a suboptimal solution is to use a greedy algorithm that iteratively adds columns of $Q$ to the seed set. Unfortunately, the straightforward greedy optimization (trying to add each item to the current seed set and computing its rectangular volume) costs $O(mL_0^2f^2)$, that often is too expensive considering typical sizes of modern recommender datasets and number of model hyperparameters. Therefore, we developed a fast algorithm with complexity $O(mL_0^2)$ that is described in the following section.

\subsection{Algorithm}
\label{sec:algorithm}

In this section, we introduce an algorithm for the selection of $L_0$ representative items using the notion of rectangular volume. At the first step, the algorithm computes the best rank-$f$ approximation of the rating matrix $R$, PureSVD (see Section~\ref{sec:puresvd} for details), and selects $f$ representative items with the pivot indices from LU-decomposition of $Q$ or with Maxvol algorithm. This seed set is further expanded by Algorithm~\ref{alg:2maxvol} in a greedy fashion: by adding new representative items one by one maximizing rectangular volume of the seed set. Further, we show that new representative item should have the maximal norm of the coefficients that represent its latent vector by the latent vectors of the current seed set. The procedure of such norm maximization is faster than the straightforward approach. At the end of this section we describe the algorithm for even faster rank-1 updating norms of coefficients.

\subsubsection{Maximization of coefficients norm} Suppose, at some step, we have already selected $L<L_0$ representative items with the indices $k\in\mathbb{N}^L$. Let $S \in \mathbb{R}^{f\times L}$ be the corresponding submatrix of $Q\in \mathbb{R}^{f\times m}$. On the next step, the algorithm selects a column $q_i \subset Q,~ i \notin k$ and adds it to the seed set: $S \leftarrow \left[S,q_i\right],$
where $[A,B]$ is an operation of horizontal concatenation of two matrices $A$ and $B$. This column should maximize the following volume:
\begin{equation}
\label{eq:volume_update}
	q_i = \argmax_{i\notin k}\text{Rectvol}\left([S,q_i]\right).
\end{equation}

Suppose $C \in \mathbb{R}^{L \times m}$ is the current matrix of coefficients from Equation (\ref{eq:coef}), and let $c_i \in \mathbb{R}^L$ be an $i$-th column of matrix $C$. Then the updated seed set from (\ref{eq:volume_update}) can be written as following:
\begin{equation}\label{formula1}
[S,q_i]=[S,Sc_i]=S[I_L,c_i].
\end{equation}
Then the volume of the seed set can be written in the following way:
\begin{equation}
\begin{split}
\label{eq:rectvol_before_coef}
    \text{Rectvol}\left([S,q_i]\right) &= \sqrt{\det \left([S,q_i][S,q_i]^\top \right)}=\\
    &=\sqrt{\det \left(SS^\top +Sc_ic_i^\top S^\top \right)}.	
\end{split}
\end{equation}
Taking into account the identity $$\det(X+AB)=\det(X)\det(I+BX^{-1}A),$$ the volume (\ref{eq:rectvol_before_coef}) can be written as following:
\begin{equation}
\label{eq:rect_det}
\text{Rectvol}\left([S,q_i]\right)=\text{Rectvol}(S)\sqrt{1+w_i},
\end{equation}
where $w_i=\Vert c_i \Vert_2^2$. Thus, the maximization of rectangular volume is equivalent to the maximization of the $l_2$-norm of the coefficients vector $c_i$, which we know only after recomputing~(\ref{eq:coef_pseudo}). Total recomputing of coefficient matrix $C$ on each iteration is faster than the straightforward approach described in Section \ref{sec:volume} and costs $O(mL_0^2fm)$. However, in the next section, we describe even faster algorithm with an efficient recomputation of the coefficients.

\subsubsection{Fast Computation of Coefficients} Since the matrix of coefficients $C$ is the least-norm solution \eqref{eq:coef_pseudo}, after adding column $q_i$ to the seed set, $C$ should be computed using Equation (\ref{formula1}):
\begin{equation}
\label{eq:recomputing_coef}
C \leftarrow [S,q_i]^\dagger Q=[I_L,c_i]^\dagger S^\dagger Q=[I_L, c_i]^\dagger C.
\end{equation}
The pseudoinverse from (\ref{eq:recomputing_coef}) can be obtained in this way:
\begin{equation*}
\begin{split}
	[I_L,c_i]^\dagger=[I_L, c_i&]^\top \left([I_L,c_i][I_L,c_i]^\top \right)^{-1}=\\
	&=\begin{bmatrix}I_L\\ c_i^\top \end{bmatrix}\left(I_L+c_ic_i^\top \right)^{-1},
\end{split}
\end{equation*}
where $\begin{bmatrix}A\\B\end{bmatrix}$ is an operation of vectical concatenation of $A$ and $B$.
The inversion in this formula can be computed by the Sherman-Morrison formula:
\begin{equation*}
\left(I_L+c_ic_i^\top \right)^{-1} = I_L-\frac{c_ic_i^\top }{1+c_i^\top c_i}.
\end{equation*}
Putting it into (\ref{eq:recomputing_coef}), we finally get the main update formula for $C$:
\begin{equation}
\label{eq:newC}
C \leftarrow \begin{bmatrix}I_L-\frac{c_ic_i^\top }{1+c_i^\top c_i}\\
\vspace{-2mm}\\ c_i^\top -\frac{c_i^\top c_ic_i^\top }{1+c_i^\top c_i}\end{bmatrix}\cdot C=\begin{bmatrix}C-\frac{c_ic_i^\top C}{1+c_i^\top c_i}\\
\vspace{-2mm}\\ \frac{c_i^\top C}{1+c_i^\top c_i}\end{bmatrix}.
\end{equation}

Recall that we should efficiently recompute norms of coefficients $w_i$. Using Equation (\ref{eq:newC}), we arrive at the following formula for the update of all norms $w_j$:
\begin{equation}
\label{eq:length_update}
w_j \leftarrow w_j-\frac{(c_i^\top c_j)^2}{1+c_i^\top c_i}.
\end{equation}
It is natural to see that coefficients norms are decreasing, because adding each new latent vector to the seed set gives more flexibility of representing all latent vectors via representative ones.

Equations (\ref{eq:newC}) and (\ref{eq:length_update}) allow to recompute $C$ and $W$ using the simple rank-$1$ update. Thus, the complexity of adding a new column into the seed set is low, what is shown in Section~\ref{sec:rectmaxvol_complexity}. The pseudocode of the algorithm is provided in Algorithm~\ref{alg:2maxvol}.

\begin{algorithm}[h]
\caption{Searching representative items using Rectangular Maxvol}
\label{alg:2maxvol}
\begin{algorithmic}[1]
\REQUIRE Rating matrix $R \in \mathbb{R}^{n \times m}$, number of representative items $L_0$, rank of decomposition $f \le L_0$
\ENSURE Indices $k\in \mathbb{N}^{L_0}$ of $L_0$ representative items
\STATE Compute rank-$f$ PureSVD of the matrix $R\approx P^\top  Q$
\STATE Get the initial square seed set: $k \leftarrow$ $L_0$ pivot indices from LU-decomposition of $Q$
\STATE $S \leftarrow Q(:,k)$
\STATE $C = S^{-1}Q$
\STATE $\forall i:\: w_i \leftarrow \Vert c_i \Vert_2^2$, where $c_i$ is the $i$-th column of $C$
\WHILE{$|k| < L_0$}
\STATE $i \leftarrow \mathrm{argmax_{i \notin k}}(w_i)$
\STATE $k \leftarrow [k, i]$
\STATE $S \leftarrow [S,q_i]$
\STATE $C \leftarrow \begin{bmatrix}C-\frac{c_ic_i^\top C}{1+c_i^\top c_i}\\ \vspace{-2mm}\\ \frac{c_i^\top C}{1+c_i^\top c_i}\end{bmatrix}$
\STATE $\forall j:\: w_j \leftarrow w_j-\frac{(c_i^\top c_j)^2}{1+c_i^\top c_i}$
\ENDWHILE
\RETURN $k$
\end{algorithmic}
\end{algorithm}

The seed sets provided by the algorithm can be used for rating elicitation and further prediction of ratings for the rest of the items, as demonstrated in Section \ref{sec:predicting}. Moreover, if the size of the seed set $L_0$ is not limited by a fixed budget, alternative stopping criteria is proposed in Section \ref{sec:bound}.

\subsection{Compelexity analysis}
\label{sec:rectmaxvol_complexity}

The proposed algorithm has two general steps: the initialization (Steps 1--5) and the iterative addition of columns or rows into the seed set (Steps 6--12). The initialization step corresponds to the LU-decomposition or Square Maxvol, which have $O(mf^2)$ complexity. Addition of one element into the seed set (Steps 7--11) requires the recomputation of the coefficients $C$ (Step 10) and lengths of coefficient vectors (Step 11). The recomputation (Step 10) requires a rank-$1$ update of the coefficients matrix $C\in \mathbb{R}^{L \times m}$ and the multiplication $c_i^\top C$, where $c_i \in \mathbb{R}^{L}$ is a column of $C$. The complexity of each of the two operations is $O(Lm)$, so the total complexity of one iteration (Steps 7--11) is $O(Lm)$. Since this procedure is iterated over $L \in \{f, ..., L_0\}$, the complexity of the loop (Step 6) is equal to $O(m(L_0^2 - f^2))$. So, in total, the complexity of Algorithm~\ref{alg:2maxvol} is $O(mL_0^2)$.

\subsection{Error Estimate}
\label{sec:bound}

\subsubsection{Analysis of Error} In this section, we theoretically analyse the estimation error of our method proposed in Section~\ref{sec:algorithm}. According to Section \ref{sec:puresvd} we have a low-rank approximation of the rating matrix $$R = P^\top  Q + \mathcal{E},$$ where $\mathcal{E} \in \mathbb{R}^{n \times m}$ is a random error matrix. On the other hand, we have RBMF approximation~(\ref{eq:r_coefs}). Let us represent its error via $\mathcal{E}$.

First of all, we have $$R(:, k) = P^\top  Q(:,k) + \mathcal{E}(:, k) = P^\top  S + \mathcal{E}(:, k).$$ Since $C=S^{\dagger}Q$ (see Section \ref{sec:predicting} for details), the RBMF approximation of $R$ can be written in the following form: $$R(:,k)C=P^\top  SS^{\dagger}Q + \mathcal{E}(:,k)C=R -\mathcal{E}+ \mathcal{E}(:,k)C,$$
what means
$$R=R(:,k)C+\mathcal{E}-\mathcal{E}(:,k)C.$$
The smaller in modulus the noise terms are, the better approximation of $R$ we have. It means that we are interested in the small values of the matrix $C$, such as the least-norm solution of (\ref{eq:coef_factor_optimization}). Further, we prove a theorem providing an approximated bound for the maximal length of $c_i$. 

\subsubsection{Upper Bound of Coefficients Norm} Similarly to Square Maxvol algorithm, a rectangular submatrix is called \textit{dominant}, if its rectangular volume does not increase by replacing one row with another one from the source matrix.

\newtheorem{theorem}{Theorem}
\begin{theorem}
\label{theorem}
Let $Q \in \mathbb{R}^{f \times m}$ be a matrix of rank $f$. Assume~$k \in \mathbb{N}^{L_0}$ is a vector of seed set element indices that produces rank-$f$ dominant submatrix of $S = Q(:,k)$, where~$S \in \mathbb{R}^{f \times L_0}$ and $m \ge L_0 \ge f$. Let $C$ be a matrix of least-norm coefficients~$C\in\mathbb{R}^{L_0\times m}$, such that $Q = SC$. Then $l_2$-norm of a column $c_i$ of $C$ for $i$ not from the seed set is bounded as: 
$$\Vert c_i \Vert_2 \le \sqrt{\frac{f}{L_0+1-f}}, \quad i \notin k.$$
\end{theorem}
\begin{proof}
Since $S$ is a dominant submatrix of the matrix $Q$, it has the maximal rectangular volume among all possible submatrices of $[S,q_i]$ with the shape $f\times L_0$. Therefore, applying Lemma~\ref{lemma} to the matrix $[S,q_i]$, we get
\begin{equation*}
\det\left([S,q_i][S,q_i]^\top \right) \le \frac{L_0+1}{L_0+1-f}\det(SS^\top ).
\end{equation*}
Using Equation (\ref{eq:rect_det}), we get:
\begin{equation*}
\begin{split}
	\Vert c_i \Vert_2^2 &= \frac{\det\left([S,q_i][S,q_i]^\top \right)}{\det(SS^\top )}-1\le \frac{f}{L_0+1-f},
\end{split}
\end{equation*}
what finishes the proof.
\end{proof}

The similar theoretical result was obtained in \cite{de2007subset},  However our proof seems to be much closely related to the notation used in our paper and in the proposed algorithm.

Theorem \ref{theorem} demonstrates that if we have an existing decomposition with the fixed rank $f$ and the size of the seed set~$L_0$ is not limited by a fixed budget it is enough to take~$L_0=2f$ items to the seed set for getting all coefficients norm less than 1. This condition of representativeness has a very natural geometric meaning: all item latent vectors are inside the ellipsoid spanned by the latent vectors from the seed set. The numerical experiments with randomly generated~$f\times m$ matrices have shown, that Algorithm \ref{alg:2maxvol} requires only~$L_0 \approx 1.2f$ rows to reach upper bound 2 for the length of each row of $C$ and only $L_0 \approx 2f$ to reach the upper bound 1 for the length of each row of $C$. So, although, our algorithm does not guarantee that the seed set submatrix is dominant, the experiment results are fully consistent with the theory.

Further, we prove the supporting lemma.

\newtheorem{lemma}{Lemma}
\begin{lemma}
\label{lemma}
Let $A \in \mathbb{R}^{N \times M}$ and $B \in \mathbb{R}^{M \times N}, M > N$.
Let~$A_{-i}$ be $N \times (M-1)$ submatrix of $A$ without $i$-th column
and~$B_{-i}$ be $(M-1) \times N$ submatrix of $B$ without $i$-th row.
Then,
$$\det(AB)\le\frac{M}{M-N}\max_{i}\left(\det(A_{-i}B_{-i})\right)$$
\end{lemma}

\begin{proof}
From the Cauchy-Binet formula we get $$\det(AB) = \sum_{k} \det A(:,k) \cdot \det B(k,:),$$ where $k\in\mathbb{N}^N$ is a vector of $N$ different indices. Since $A_{-i}$ contains all columns of $A$ except $i$-th column, then $A(:,k)$ is a submatrix of $A_{-i}$ for any $i \notin k$. Since $k$ consists of~$N$ different numbers, we have $M-N$ different $i$, such that~$A(:,k)$ is a submatrix of $A_{-i}$. The same is true for the matrix $B$. So get $$\sum_{i=1}^M \det(A_{-i}B_{-i}) = (M-N)\det(AB)$$ applying Cauchy-Binet formula to each summand. Therefore, $$\det(AB) = \frac{1}{M-N}\sum_{i=1}^M\det(A_{-i}B_{-i}),$$ what finishes the proof. 
\end{proof}

\section{Experiments}
\label{sec:experiments}

The proposed experiments\footnote{The source code is available here: \url{https://bitbucket.org/muxas/rectmaxvol_recommender}} compare two algorithms: Square Maxvol based (our primary baseline) and Rectangular Maxvol based (Section \ref{sec:proposed_method}). Other competitors have either an infeasible computational complexity (see Section \ref{sec:related} for details) or have a lower quality than our baseline, as it is shown in \cite{liu2011wisdom} (we reproduced the conclusions from \cite{liu2011wisdom} but they are not demonstrated here due to the lack of space). Moreover, it is important to note that the experiments in \cite{liu2011wisdom} used smaller versions of the datasets. Therefore, the performance of Square Maxvol on the extended datasets is different from that reported in \cite{liu2011wisdom}.

\subsection{Datasets.}
We used two popular publicly available datasets in our experiments. T first one is the Movielens dataset\footnote{\url{http://grouplens.org/datasets/movielens/}} which contains 20,000,263 ratings of 26,744 movies from 138,493 users. The analysis of the older and smaller version of this dataset is provided in~\cite{miller2003movielens}.
The second one is the Netflix dataset\footnote{\url{http://www.netflixprize.com/}}. It contains 100,480,507 ratings of 17,770 movies from 480,189 users. The description of the dataset and the competition can be found in \cite{bennett2007netflix}.
The rating matrix $R$ was formed in the same way as in \cite{liu2011wisdom}.

\subsection{Evaluation Protocol.}
Our evaluation pipeline for the comparison of the rating elicitation algorithms is similar to the one introduced in \cite{liu2011wisdom}. All our experiments are provided for both the user and the item cold start problems. However, without loss of generality, this section describes the evaluation protocol for the user cold start problem only. The item cold start problem can be evaluated in the same way after the transposition of the rating matrix.

We evaluate the algorithms for selecting representatives by the assessing the quality of the recommendations recovered after the acquisition of the actual ratings of the representatives, what can be done as shown in Section \ref{sec:predicting}. Note that users may not have ratings for the items from the seed set: if user $u$ was asked to rate item $i$ with unknown rating, then, according to PureSVD model, $r_{ui}$ is set to $0$.
In case of the user cold start problem, all users are randomly divided into 5 folds of equal size, and the experiments are repeated 5 times, assuming that one part is a test set with cold users and the other four parts form the train set and the validation set contain warm users. Analogically, in case of the item cold start, all items were divided into 5 folds.
	
Pointwise quality measures are easy to be optimized directly, but they are not very suitable for recommendation quality evaluation, because the goal of a recommender system is not to predict particular rating values, but to predict the most relevant recommendations that should be shown to the user. That is why, we use ranking measures to evaluate all methods~\cite{gunawardana2009survey}. For evaluation, we divided all items for every user into relevant and irrelevant ones, as it was done in the baseline paper \cite{liu2011wisdom}.

One of the most popular and interpretable ranking measures for the recommender systems evaluation are Precision@$k$ and Recall@$k$ \cite{cremonesi2010performance} that measure the quality of top-$k$ recommendations in terms of their relevance. More formally, Precision@$k$ is the fraction of relevant items among the top-$k$ recommendations. Recall@$k$ is the fraction of relevant items from the top $k$ among all relevant items. Our final evaluation measures were computed by averaging Precision@$k$ and Recall@$k$ over all users in the test set.
Note that in the case of the item cold start problem, Precision@$k$ and Recall@$k$ are computed on the transposed rating matrix $R$. Moreover, following the methodology from \cite{liu2011wisdom}, we compare algorithms in terms of coverage and diversity.

\subsection{Results of Experiments.}
As we mentioned in Section~\ref{sec:predicting}, there are two different ways to compute the coefficients for representing the hidden ratings via the ratings from a seed set. The first one is to compute them via the low-rank factors, as shown in Equation (\ref{eq:coef_pseudo}). The second one is to compute them via the source rating matrix~$R$, as shown in Equation (\ref{eq:coef_matrix}). Our experiments show that the second approach demonstrates the significantly better quality. Therefore, we use this method in all our experiments.

We processed experiments for the seed set sizes from 5 to~100 with a step of 5. These computations become possible for such dense grid of parameters, because of the high computational efficiency of our algorithm (see Section \ref{sec:related}). The average computational time of Rectangular Maxvol on the datasets is $1.12$ seconds (Intel Xeon CPU 2.00GHz, 256Gb RAM). The average computational time of Square Maxvol is almost the same what confirms the theoretical complexity analysis.

In the case of Rectangular Maxvol, for every size of the seed set, we used the rank that gives the best performance on a separate evaluation set. Fig.~\ref{fig:major} demonstrates the superiority of our approach over the ordinal Square Maxvol for all cold start problems types (user and item) and for both datasets. Moreover, it can be seen from the magnitudes of the differences that Rectangular Maxvol gives much more stable results that the square one. The same conclusions can be made for any combination of Precision/Recall, $k$ and seed set sizes, but they are not demonstrated here due to the lack of space.

As mentioned above, Rectangular Maxvol used the optimal rank value in our experiments. Fig.~\ref{fig:opt_rank} demonstrates the averaged optimal rank over all experiments for all datasets and for all cold start problem types. It is easy to see that, in each case, the required optimal rank is significantly smaller than the corresponding size of the seed set. This unequivocally confirms that the rectangular generalization of the square maximal-volume concept makes a great sense. Moreover, since Rectangular Maxvol requires a smaller rank of the rating matrix factorization, it is more computationally and memory efficient.

\begin{figure*}[!htbp]
\begin{multicols}{2}
\hfill
\hspace{-9mm}
{\resizebox{0.5\textwidth}{!}{\input{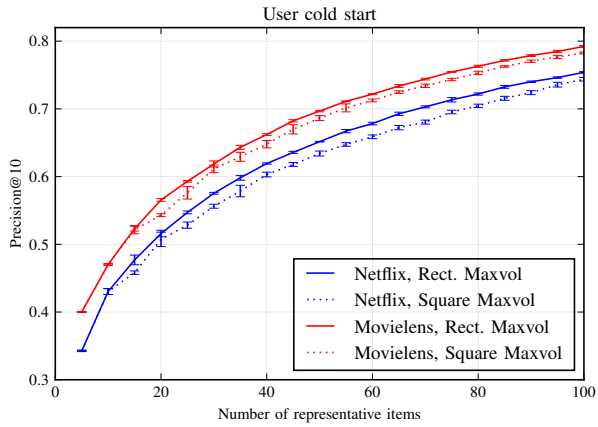}}
\label{fig:1}
\hfill
\hspace{-6mm}
\resizebox{0.5\textwidth}{!}{\input{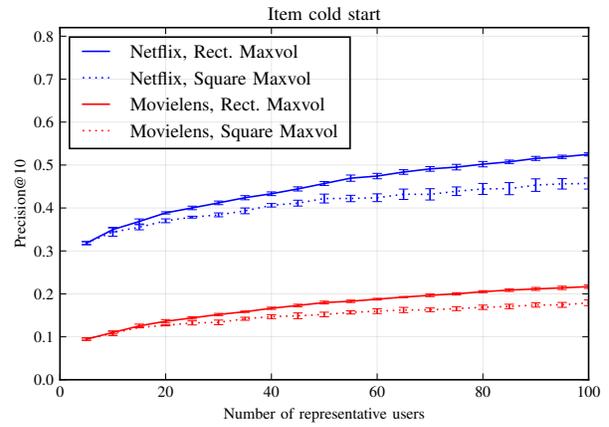}}
\label{fig:2}
\hfill
}
\hfill
\end{multicols}
\caption{Precision@10 dependence on the size of the seed set. The comparison of Square Maxvol and Rectangular Maxvol. The errorbars indicate $\sigma$ deviation.}
\label{fig:major}
\end{figure*}

\begin{figure*}[!htbp]
\begin{multicols}{2}
\hfill
\hspace{-9mm}
\resizebox{0.5\textwidth}{!}{\input{figures/paper_rank_size.pgf}}
\caption{Optimal rank dependence on the size of the seed set.}
\hfill
\label{fig:opt_rank}
\hspace{-7mm}
\resizebox{0.5\textwidth}{!}{\input{figures/coverage_diversity.pgf}}
\caption{Coverage and diversity of the Netflix seed set items.}
\hfill
\label{fig:coverage_diversity}
\end{multicols}
\end{figure*}

On Fig. \ref{fig:coverage_diversity}, we can see that the coverage and diversity measures \cite{liu2011wisdom} of the representative Netflix items selected by Rectangular Maxvol are higher than the measures of Square Maxvol. The cases of representative users and Movielens dataset lead to the same results, but the corresponding figures are not demonstrated here due to the lack of space.

In the end, it is interesting to analyse the behaviour of the automatic stopping criterion that adds objects into the seed set until all latent vectors are covered by the ellipsoid spanned by the latent vectors of the representatives. The experiments show that increasing the rank results in a quality fall in the case of representative users and the ranks higher than 50, what means an overfitting of PureSVD. In case of the representative items, the quality becomes almost constant starting from the same ranks.

\section{Related work}
\label{sec:related}

\subsection{Collaborative Filtering.} 
Conventional CF methods do not analyse any domain-specific context of users/items~\cite{schafer2007collaborative}, such as explicit user and item profiles, items' text descriptions or social relations between users. Therefore, they are domain- and data-independent and can be applied to a wide range of tasks, what is their major advantage. As shown in~\cite{cremonesi2010performance}, CF approaches based on a factorization have high accuracy for the majority of datasets. While a particular choice of a factorization algorithm is not essential for our approach to the cold start problem, our methodology is based on the PureSVD, which performs better than other popular methods such as SVD++ \cite{cremonesi2010performance}.

\subsection{Scoring Rating Elicitation Methods.} 
The simplest methods for the seed set selection rank users or items by some ad-hoc score which shows how representative they are and take the top-$k$ ranked entities as a seed set~\cite{rashid2002getting, rashid2008learning, zhang2014less, zhang2013semi}.
An obvious drawback of such methods that is avoided in our approach is that these elements are taken from the seed set independently and diversity of the selected elements is limited~\cite{golbandi2010bootstrapping}. Further in this section, we overview the methods that aim on a selection of a diverse seed set and that have better performance. This is why we do not use the scoring methods in our experiments.

\subsection{GreedyExtend.}
Among them, the most straightforward method is the GreedyExtend approach \cite{golbandi2010bootstrapping}.
Unfortunately, the brute force manner of GreedyExtend implies very high computational costs. Hence, it is hardly scalable, in contrast to the approaches that are empirically compared in this paper. This method greedily adds the item $i$ to the current seed set of indices~$k\in \mathbb{N}^L$ that maximizes the target quality measure. The search of the best $i$ is computed in a brute force manner, i.e. the algorithm iteratively adds the best item into the seed set:~$k \leftarrow [k, i]$, where $i=\argmin_{i' \notin k}\mathcal F([k,i'])$ and $\mathcal{F}([k,i'])$ is the quality measure of recommendations generated using the seed set indices $[k,i']$. The authors of this method reported the results only for an approach that uses similarities of items to predict the ratings via the seed set ratings. More effective~\cite{anava2015budget} linear approach described in Section \ref{sec:predicting} costs~$O(Lnm)$, where~$L=|k|$. At each step, the least squares solution is computed for almost all items, i.e.~$O(m)$ times. Since the algorithm has $L_0$ such steps, the total complexity is~$O(L_0^2nm^2)$ (more than $10^{16}$ operations for the Netflix dataset and the seed set size $L_0=10$). Therefore, we do not use this method in our experiments.

\subsection{Backward Greedy Selection.}
Another class of methods of searching for diverse representatives is based on the factorization of the rating matrix. Since the selection of user or item representatives is equivalent to selecting a submatrix of the corresponding factor, these algorithms seek for the submatrix that maximizes some criterion. One such approach, called Backward Greedy Selection~\cite{anava2015budget}, solves only the item cold start problem, but not the user one. This method is based on the techniques for transductive experimental design introduced in \cite{yu2006active}. To get the seed set, it greedily removes users from a source user set in order to get a good seed set minimizing the value $\text{Trace}\left((SS^\top )^{-1}\right),$ where~$S \in \mathbb{R}^{f \times L}$ is a submatrix in the items' factor~$Q \in \mathbb{R}^{f \times m}$ of a rank-$f$ decomposition. Each deletion of an item requires iterative look up of all the items in the data, where each iteration costs $O(f^2L)$. So, one deletion takes $O(f^2Lm)$ operations. Assuming that $L_0\ll m$, the whole procedure takes~$O(f^2m^3)$ operations, which is too expensive to be computed on real world datasets (the authors have selected a small subset of users to perform their evaluation). Therefore, we do not use this method in our experiments.

\subsection{Representative Based Matrix Factorization.}
The method presented in \cite{liu2011wisdom}, called Representative Based Matrix Factorization (RBMF), takes the diversity into account as well. It uses maximal-volume concept and the Maxvol algorithm \cite{goreinov2010find} for searching the most representative rows or columns in the factors of a CF factorization. This approach is highly efficient and more accurate than all ad-hoc competitors, but it also has one important limitation. It must use the same rank of factorization as the desired number of representative users or items for the seed set. The algorithm proposed in our paper is a generalization of Maxvol that allows to use different rank values. It often leads to a better recommendation accuracy, as shown in Section \ref{sec:experiments}.

\subsection{Complexity analysis.}
Let us overview the computational complexity of the proposed Rectangular Maxvol and its competitors. Some of these methods use low-rank factorizations of the matrix, whose detailed complexity analysis is provided in~\cite{golub2012matrix}. However, as this is not a key point of our work, we neglect the computational cost of factorizations in the further analysis, because  it is same for all rating elicitation algorithms and usually is previously computed for the warm CF method. The summary of the complexity analysis is shown in Table~\ref{tab:complexity}. The detailed complexity analysis of Square Maxvol and Rectangular Maxvol is provided in Sections \ref{sec:maximal} and \ref{sec:rectmaxvol_complexity} respectively.

\begin{table}[t]
  \centering
  \begin{tabular}{|l|l|l|}
    \hline
    \textbf{Algorithm} & \textbf{Complexity} \\
    \hline
	\hline
	Square Maxvol & $O(mL_0^2)$ \\
	\hline
	Rectangular Maxvol & $O(mL_0^2)$ \\
	\hline
	GreedyExtend & $O(m^2nL_0^2)$ \\
	\hline
	Backward Greedy & $O(m^3f^2)$\\
    \hline
  \end{tabular}
	\vspace{2mm} 
  \caption{Complexity of the algorithms}
  \label{tab:complexity}
\end{table}

\subsection{Cold Start Problem}
Apart from rating elicitation methods, there were also different approaches to cold start problem proposed in the literature. Additional context information (e.g., category labels \cite{zhang2015dualds} or all available metadata \cite{barjasteh2015cold}) may be used. Moreover, there is a class of methods that use adaptive tree-based questionnaires to acquire the initial information about new users \cite{golbandi2011adaptive, graus2015improving, karimi2014improved, karimi2013factorized, sun2013learning,  zhou2011functional}. Moreover, the cold start problem can be viewed from the exploration-exploitation trade-off point of view \cite{ aharon2015excuseme, zhao2013interactive}. The methods from \cite{cremonesi2012user, kluver2014evaluating} analyse the performance of CF methods w.r.t.\ the number of known ratings for a user.

\subsection{Maximal-Volume Concept}
The maximal-volume concept, originally described in the field of low-rank approximation of matrices~\cite{goreinov2001maximal}, provides an approach for a matrix approximation in a pseudo-skeleton form, which is a product of matrices formed by columns or rows of the source matrix.
The algorithm, called Maxvol~\cite{goreinov2010find}, allows to efficiently find a well-conditioned submatrix with a high enough volume for building such an approximation. Maximal volume submatrices are useful not only for low-rank approximations, but also in wireless communications~\cite{wang2010global}, preconditioning of overdetermined systems \cite{arioli2014preconditioning}, tensor decompositions \cite{oseledets2010tt}, and recommender systems \cite{liu2011wisdom}. Our generalization of the maximal-volume concept to rectangular case offers additional degrees of freedom, what is potentially useful in any of these areas.

\section{Conclusions}
\label{sec:conclusions}

In our paper, we overviewed the existing approaches for the rating elicitation and introduced the efficient algorithm based on the definition of rectangular matrix volume. Moreover, in order to demonstrate the superiority of the proposed method, we provided the analytical and experimental comparison to the existing approaches. It seems to be an interesting direction of future work to apply the proposed framework to building tree-based cold-start questionnaires in recommender systems.

Another interesting direction for future work is to join approaches from two classes: based on the maximal-volume concept and based on optimal design criteria. They historically came from absolutely different fields: from computational lineal algebra and from statistical experimental analysis respectively. Although all these methods are very similar from the mathematical point of view, it seems quite interesting to explore their similarities and differences.

\section{Acknowledgments}

Work on problem setting and numerical examples was supported by Russian
Science Foundation grant 14-11-00659. Work on theoretical estimations of approximation error and practical algorithm was supported by Russian Foundation for Basic Research 16-31-00351 mol\_a. Also we thank Evgeny Frolov for helpful discussions.

\bibliographystyle{./IEEEtran}

\input{cold_start_rect_maxvol_icdm.bbl}

\end{document}

%% file: figures/paper_rank_size.pgf
%% Creator: Matplotlib, PGF backend
%%
%% To include the figure in your LaTeX document, write
%%   \input{<filename>.pgf}
%%
%% Make sure the required packages are loaded in your preamble
%%   \usepackage{pgf}
%%
%% Figures using additional raster images can only be included by \input if
%% they are in the same directory as the main LaTeX file. For loading figures
%% from other directories you can use the `import` package
%% and then include the figures with
%%   \import{<path to file>}{<filename>.pgf}
%%
%% Matplotlib used the following preamble
%%   \usepackage{fontspec}
%%
\begingroup%
\makeatletter%
\begin{pgfpicture}%
\pgfpathrectangle{\pgfpointorigin}{\pgfqpoint{6.000000in}{4.000000in}}%
\pgfusepath{use as bounding box, clip}%
\begin{pgfscope}%
\pgfsetbuttcap%
\pgfsetmiterjoin%
\definecolor{currentfill}{rgb}{1.000000,1.000000,1.000000}%
\pgfsetfillcolor{currentfill}%
\pgfsetlinewidth{0.000000pt}%
\definecolor{currentstroke}{rgb}{1.000000,1.000000,1.000000}%
\pgfsetstrokecolor{currentstroke}%
\pgfsetdash{}{0pt}%
\pgfpathmoveto{\pgfqpoint{0.000000in}{0.000000in}}%
\pgfpathlineto{\pgfqpoint{6.000000in}{0.000000in}}%
\pgfpathlineto{\pgfqpoint{6.000000in}{4.000000in}}%
\pgfpathlineto{\pgfqpoint{0.000000in}{4.000000in}}%
\pgfpathclose%
\pgfusepath{fill}%
\end{pgfscope}%
\begin{pgfscope}%
\pgfsetbuttcap%
\pgfsetmiterjoin%
\pgfsetlinewidth{0.000000pt}%
\definecolor{currentstroke}{rgb}{0.000000,0.000000,0.000000}%
\pgfsetstrokecolor{currentstroke}%
\pgfsetstrokeopacity{0.000000}%
\pgfsetdash{}{0pt}%
\pgfpathmoveto{\pgfqpoint{0.750000in}{0.500000in}}%
\pgfpathlineto{\pgfqpoint{5.400000in}{0.500000in}}%
\pgfpathlineto{\pgfqpoint{5.400000in}{3.600000in}}%
\pgfpathlineto{\pgfqpoint{0.750000in}{3.600000in}}%
\pgfpathclose%
\pgfusepath{}%
\end{pgfscope}%
\begin{pgfscope}%
\pgfpathrectangle{\pgfqpoint{0.750000in}{0.500000in}}{\pgfqpoint{4.650000in}{3.100000in}} %
\pgfusepath{clip}%
\pgfsetrectcap%
\pgfsetroundjoin%
\pgfsetlinewidth{1.003750pt}%
\definecolor{currentstroke}{rgb}{0.000000,0.000000,1.000000}%
\pgfsetstrokecolor{currentstroke}%
\pgfsetdash{}{0pt}%
\pgfpathmoveto{\pgfqpoint{0.982500in}{0.500000in}}%
\pgfpathlineto{\pgfqpoint{1.215000in}{0.593000in}}%
\pgfpathlineto{\pgfqpoint{1.447500in}{0.562000in}}%
\pgfpathlineto{\pgfqpoint{1.680000in}{0.531000in}}%
\pgfpathlineto{\pgfqpoint{1.912500in}{0.531000in}}%
\pgfpathlineto{\pgfqpoint{2.145000in}{0.531000in}}%
\pgfpathlineto{\pgfqpoint{2.377500in}{0.531000in}}%
\pgfpathlineto{\pgfqpoint{2.610000in}{0.748000in}}%
\pgfpathlineto{\pgfqpoint{2.842500in}{0.748000in}}%
\pgfpathlineto{\pgfqpoint{3.075000in}{0.810000in}}%
\pgfpathlineto{\pgfqpoint{3.307500in}{0.810000in}}%
\pgfpathlineto{\pgfqpoint{3.540000in}{0.810000in}}%
\pgfpathlineto{\pgfqpoint{3.772500in}{0.872000in}}%
\pgfpathlineto{\pgfqpoint{4.005000in}{0.903000in}}%
\pgfpathlineto{\pgfqpoint{4.237500in}{0.903000in}}%
\pgfpathlineto{\pgfqpoint{4.470000in}{0.903000in}}%
\pgfpathlineto{\pgfqpoint{4.702500in}{0.872000in}}%
\pgfpathlineto{\pgfqpoint{4.935000in}{0.965000in}}%
\pgfpathlineto{\pgfqpoint{5.167500in}{1.120000in}}%
\pgfpathlineto{\pgfqpoint{5.400000in}{1.399000in}}%
\pgfusepath{stroke}%
\end{pgfscope}%
\begin{pgfscope}%
\pgfpathrectangle{\pgfqpoint{0.750000in}{0.500000in}}{\pgfqpoint{4.650000in}{3.100000in}} %
\pgfusepath{clip}%
\pgfsetrectcap%
\pgfsetroundjoin%
\pgfsetlinewidth{1.003750pt}%
\definecolor{currentstroke}{rgb}{0.000000,0.500000,0.000000}%
\pgfsetstrokecolor{currentstroke}%
\pgfsetdash{}{0pt}%
\pgfpathmoveto{\pgfqpoint{0.982500in}{0.500000in}}%
\pgfpathlineto{\pgfqpoint{1.215000in}{0.593000in}}%
\pgfpathlineto{\pgfqpoint{1.447500in}{0.655000in}}%
\pgfpathlineto{\pgfqpoint{1.680000in}{0.686000in}}%
\pgfpathlineto{\pgfqpoint{1.912500in}{0.624000in}}%
\pgfpathlineto{\pgfqpoint{2.145000in}{0.686000in}}%
\pgfpathlineto{\pgfqpoint{2.377500in}{0.686000in}}%
\pgfpathlineto{\pgfqpoint{2.610000in}{0.655000in}}%
\pgfpathlineto{\pgfqpoint{2.842500in}{0.748000in}}%
\pgfpathlineto{\pgfqpoint{3.075000in}{0.655000in}}%
\pgfpathlineto{\pgfqpoint{3.307500in}{0.748000in}}%
\pgfpathlineto{\pgfqpoint{3.540000in}{0.748000in}}%
\pgfpathlineto{\pgfqpoint{3.772500in}{0.841000in}}%
\pgfpathlineto{\pgfqpoint{4.005000in}{0.903000in}}%
\pgfpathlineto{\pgfqpoint{4.237500in}{0.872000in}}%
\pgfpathlineto{\pgfqpoint{4.470000in}{0.872000in}}%
\pgfpathlineto{\pgfqpoint{4.702500in}{0.810000in}}%
\pgfpathlineto{\pgfqpoint{4.935000in}{0.841000in}}%
\pgfpathlineto{\pgfqpoint{5.167500in}{0.872000in}}%
\pgfpathlineto{\pgfqpoint{5.400000in}{1.027000in}}%
\pgfusepath{stroke}%
\end{pgfscope}%
\begin{pgfscope}%
\pgfpathrectangle{\pgfqpoint{0.750000in}{0.500000in}}{\pgfqpoint{4.650000in}{3.100000in}} %
\pgfusepath{clip}%
\pgfsetrectcap%
\pgfsetroundjoin%
\pgfsetlinewidth{1.003750pt}%
\definecolor{currentstroke}{rgb}{1.000000,0.000000,0.000000}%
\pgfsetstrokecolor{currentstroke}%
\pgfsetdash{}{0pt}%
\pgfpathmoveto{\pgfqpoint{0.982500in}{0.500000in}}%
\pgfpathlineto{\pgfqpoint{1.215000in}{0.616250in}}%
\pgfpathlineto{\pgfqpoint{1.447500in}{0.538750in}}%
\pgfpathlineto{\pgfqpoint{1.680000in}{0.616250in}}%
\pgfpathlineto{\pgfqpoint{1.912500in}{0.732500in}}%
\pgfpathlineto{\pgfqpoint{2.145000in}{0.732500in}}%
\pgfpathlineto{\pgfqpoint{2.377500in}{0.848750in}}%
\pgfpathlineto{\pgfqpoint{2.610000in}{0.655000in}}%
\pgfpathlineto{\pgfqpoint{2.842500in}{0.655000in}}%
\pgfpathlineto{\pgfqpoint{3.075000in}{0.655000in}}%
\pgfpathlineto{\pgfqpoint{3.307500in}{0.655000in}}%
\pgfpathlineto{\pgfqpoint{3.540000in}{0.616250in}}%
\pgfpathlineto{\pgfqpoint{3.772500in}{0.655000in}}%
\pgfpathlineto{\pgfqpoint{4.005000in}{0.655000in}}%
\pgfpathlineto{\pgfqpoint{4.237500in}{0.655000in}}%
\pgfpathlineto{\pgfqpoint{4.470000in}{1.003750in}}%
\pgfpathlineto{\pgfqpoint{4.702500in}{0.655000in}}%
\pgfpathlineto{\pgfqpoint{4.935000in}{0.655000in}}%
\pgfpathlineto{\pgfqpoint{5.167500in}{0.655000in}}%
\pgfpathlineto{\pgfqpoint{5.400000in}{0.655000in}}%
\pgfusepath{stroke}%
\end{pgfscope}%
\begin{pgfscope}%
\pgfpathrectangle{\pgfqpoint{0.750000in}{0.500000in}}{\pgfqpoint{4.650000in}{3.100000in}} %
\pgfusepath{clip}%
\pgfsetrectcap%
\pgfsetroundjoin%
\pgfsetlinewidth{1.003750pt}%
\definecolor{currentstroke}{rgb}{0.000000,0.750000,0.750000}%
\pgfsetstrokecolor{currentstroke}%
\pgfsetdash{}{0pt}%
\pgfpathmoveto{\pgfqpoint{0.982500in}{0.500000in}}%
\pgfpathlineto{\pgfqpoint{1.215000in}{0.577500in}}%
\pgfpathlineto{\pgfqpoint{1.447500in}{0.771250in}}%
\pgfpathlineto{\pgfqpoint{1.680000in}{0.810000in}}%
\pgfpathlineto{\pgfqpoint{1.912500in}{0.771250in}}%
\pgfpathlineto{\pgfqpoint{2.145000in}{0.926250in}}%
\pgfpathlineto{\pgfqpoint{2.377500in}{0.926250in}}%
\pgfpathlineto{\pgfqpoint{2.610000in}{1.003750in}}%
\pgfpathlineto{\pgfqpoint{2.842500in}{1.120000in}}%
\pgfpathlineto{\pgfqpoint{3.075000in}{1.120000in}}%
\pgfpathlineto{\pgfqpoint{3.307500in}{0.655000in}}%
\pgfpathlineto{\pgfqpoint{3.540000in}{0.655000in}}%
\pgfpathlineto{\pgfqpoint{3.772500in}{0.732500in}}%
\pgfpathlineto{\pgfqpoint{4.005000in}{0.887500in}}%
\pgfpathlineto{\pgfqpoint{4.237500in}{0.810000in}}%
\pgfpathlineto{\pgfqpoint{4.470000in}{0.887500in}}%
\pgfpathlineto{\pgfqpoint{4.702500in}{0.887500in}}%
\pgfpathlineto{\pgfqpoint{4.935000in}{0.926250in}}%
\pgfpathlineto{\pgfqpoint{5.167500in}{1.158750in}}%
\pgfpathlineto{\pgfqpoint{5.400000in}{0.887500in}}%
\pgfusepath{stroke}%
\end{pgfscope}%
\begin{pgfscope}%
\pgfpathrectangle{\pgfqpoint{0.750000in}{0.500000in}}{\pgfqpoint{4.650000in}{3.100000in}} %
\pgfusepath{clip}%
\pgfsetbuttcap%
\pgfsetroundjoin%
\pgfsetlinewidth{1.003750pt}%
\definecolor{currentstroke}{rgb}{0.750000,0.000000,0.750000}%
\pgfsetstrokecolor{currentstroke}%
\pgfsetdash{{1.000000pt}{3.000000pt}}{0.000000pt}%
\pgfpathmoveto{\pgfqpoint{0.796500in}{0.531000in}}%
\pgfpathlineto{\pgfqpoint{0.843000in}{0.562000in}}%
\pgfpathlineto{\pgfqpoint{0.889500in}{0.593000in}}%
\pgfpathlineto{\pgfqpoint{0.936000in}{0.624000in}}%
\pgfpathlineto{\pgfqpoint{0.982500in}{0.655000in}}%
\pgfpathlineto{\pgfqpoint{1.029000in}{0.686000in}}%
\pgfpathlineto{\pgfqpoint{1.075500in}{0.717000in}}%
\pgfpathlineto{\pgfqpoint{1.122000in}{0.748000in}}%
\pgfpathlineto{\pgfqpoint{1.168500in}{0.779000in}}%
\pgfpathlineto{\pgfqpoint{1.215000in}{0.810000in}}%
\pgfpathlineto{\pgfqpoint{1.261500in}{0.841000in}}%
\pgfpathlineto{\pgfqpoint{1.308000in}{0.872000in}}%
\pgfpathlineto{\pgfqpoint{1.354500in}{0.903000in}}%
\pgfpathlineto{\pgfqpoint{1.401000in}{0.934000in}}%
\pgfpathlineto{\pgfqpoint{1.447500in}{0.965000in}}%
\pgfpathlineto{\pgfqpoint{1.494000in}{0.996000in}}%
\pgfpathlineto{\pgfqpoint{1.540500in}{1.027000in}}%
\pgfpathlineto{\pgfqpoint{1.587000in}{1.058000in}}%
\pgfpathlineto{\pgfqpoint{1.633500in}{1.089000in}}%
\pgfpathlineto{\pgfqpoint{1.680000in}{1.120000in}}%
\pgfpathlineto{\pgfqpoint{1.726500in}{1.151000in}}%
\pgfpathlineto{\pgfqpoint{1.773000in}{1.182000in}}%
\pgfpathlineto{\pgfqpoint{1.819500in}{1.213000in}}%
\pgfpathlineto{\pgfqpoint{1.866000in}{1.244000in}}%
\pgfpathlineto{\pgfqpoint{1.912500in}{1.275000in}}%
\pgfpathlineto{\pgfqpoint{1.959000in}{1.306000in}}%
\pgfpathlineto{\pgfqpoint{2.005500in}{1.337000in}}%
\pgfpathlineto{\pgfqpoint{2.052000in}{1.368000in}}%
\pgfpathlineto{\pgfqpoint{2.098500in}{1.399000in}}%
\pgfpathlineto{\pgfqpoint{2.145000in}{1.430000in}}%
\pgfpathlineto{\pgfqpoint{2.191500in}{1.461000in}}%
\pgfpathlineto{\pgfqpoint{2.238000in}{1.492000in}}%
\pgfpathlineto{\pgfqpoint{2.284500in}{1.523000in}}%
\pgfpathlineto{\pgfqpoint{2.331000in}{1.554000in}}%
\pgfpathlineto{\pgfqpoint{2.377500in}{1.585000in}}%
\pgfpathlineto{\pgfqpoint{2.424000in}{1.616000in}}%
\pgfpathlineto{\pgfqpoint{2.470500in}{1.647000in}}%
\pgfpathlineto{\pgfqpoint{2.517000in}{1.678000in}}%
\pgfpathlineto{\pgfqpoint{2.563500in}{1.709000in}}%
\pgfpathlineto{\pgfqpoint{2.610000in}{1.740000in}}%
\pgfpathlineto{\pgfqpoint{2.656500in}{1.771000in}}%
\pgfpathlineto{\pgfqpoint{2.703000in}{1.802000in}}%
\pgfpathlineto{\pgfqpoint{2.749500in}{1.833000in}}%
\pgfpathlineto{\pgfqpoint{2.796000in}{1.864000in}}%
\pgfpathlineto{\pgfqpoint{2.842500in}{1.895000in}}%
\pgfpathlineto{\pgfqpoint{2.889000in}{1.926000in}}%
\pgfpathlineto{\pgfqpoint{2.935500in}{1.957000in}}%
\pgfpathlineto{\pgfqpoint{2.982000in}{1.988000in}}%
\pgfpathlineto{\pgfqpoint{3.028500in}{2.019000in}}%
\pgfpathlineto{\pgfqpoint{3.075000in}{2.050000in}}%
\pgfpathlineto{\pgfqpoint{3.121500in}{2.081000in}}%
\pgfpathlineto{\pgfqpoint{3.168000in}{2.112000in}}%
\pgfpathlineto{\pgfqpoint{3.214500in}{2.143000in}}%
\pgfpathlineto{\pgfqpoint{3.261000in}{2.174000in}}%
\pgfpathlineto{\pgfqpoint{3.307500in}{2.205000in}}%
\pgfpathlineto{\pgfqpoint{3.354000in}{2.236000in}}%
\pgfpathlineto{\pgfqpoint{3.400500in}{2.267000in}}%
\pgfpathlineto{\pgfqpoint{3.447000in}{2.298000in}}%
\pgfpathlineto{\pgfqpoint{3.493500in}{2.329000in}}%
\pgfpathlineto{\pgfqpoint{3.540000in}{2.360000in}}%
\pgfpathlineto{\pgfqpoint{3.586500in}{2.391000in}}%
\pgfpathlineto{\pgfqpoint{3.633000in}{2.422000in}}%
\pgfpathlineto{\pgfqpoint{3.679500in}{2.453000in}}%
\pgfpathlineto{\pgfqpoint{3.726000in}{2.484000in}}%
\pgfpathlineto{\pgfqpoint{3.772500in}{2.515000in}}%
\pgfpathlineto{\pgfqpoint{3.819000in}{2.546000in}}%
\pgfpathlineto{\pgfqpoint{3.865500in}{2.577000in}}%
\pgfpathlineto{\pgfqpoint{3.912000in}{2.608000in}}%
\pgfpathlineto{\pgfqpoint{3.958500in}{2.639000in}}%
\pgfpathlineto{\pgfqpoint{4.005000in}{2.670000in}}%
\pgfpathlineto{\pgfqpoint{4.051500in}{2.701000in}}%
\pgfpathlineto{\pgfqpoint{4.098000in}{2.732000in}}%
\pgfpathlineto{\pgfqpoint{4.144500in}{2.763000in}}%
\pgfpathlineto{\pgfqpoint{4.191000in}{2.794000in}}%
\pgfpathlineto{\pgfqpoint{4.237500in}{2.825000in}}%
\pgfpathlineto{\pgfqpoint{4.284000in}{2.856000in}}%
\pgfpathlineto{\pgfqpoint{4.330500in}{2.887000in}}%
\pgfpathlineto{\pgfqpoint{4.377000in}{2.918000in}}%
\pgfpathlineto{\pgfqpoint{4.423500in}{2.949000in}}%
\pgfpathlineto{\pgfqpoint{4.470000in}{2.980000in}}%
\pgfpathlineto{\pgfqpoint{4.516500in}{3.011000in}}%
\pgfpathlineto{\pgfqpoint{4.563000in}{3.042000in}}%
\pgfpathlineto{\pgfqpoint{4.609500in}{3.073000in}}%
\pgfpathlineto{\pgfqpoint{4.656000in}{3.104000in}}%
\pgfpathlineto{\pgfqpoint{4.702500in}{3.135000in}}%
\pgfpathlineto{\pgfqpoint{4.749000in}{3.166000in}}%
\pgfpathlineto{\pgfqpoint{4.795500in}{3.197000in}}%
\pgfpathlineto{\pgfqpoint{4.842000in}{3.228000in}}%
\pgfpathlineto{\pgfqpoint{4.888500in}{3.259000in}}%
\pgfpathlineto{\pgfqpoint{4.935000in}{3.290000in}}%
\pgfpathlineto{\pgfqpoint{4.981500in}{3.321000in}}%
\pgfpathlineto{\pgfqpoint{5.028000in}{3.352000in}}%
\pgfpathlineto{\pgfqpoint{5.074500in}{3.383000in}}%
\pgfpathlineto{\pgfqpoint{5.121000in}{3.414000in}}%
\pgfpathlineto{\pgfqpoint{5.167500in}{3.445000in}}%
\pgfpathlineto{\pgfqpoint{5.214000in}{3.476000in}}%
\pgfpathlineto{\pgfqpoint{5.260500in}{3.507000in}}%
\pgfpathlineto{\pgfqpoint{5.307000in}{3.538000in}}%
\pgfpathlineto{\pgfqpoint{5.353500in}{3.569000in}}%
\pgfpathlineto{\pgfqpoint{5.400000in}{3.600000in}}%
\pgfusepath{stroke}%
\end{pgfscope}%
\begin{pgfscope}%
\pgfsetrectcap%
\pgfsetmiterjoin%
\pgfsetlinewidth{1.003750pt}%
\definecolor{currentstroke}{rgb}{0.000000,0.000000,0.000000}%
\pgfsetstrokecolor{currentstroke}%
\pgfsetdash{}{0pt}%
\pgfpathmoveto{\pgfqpoint{0.750000in}{3.600000in}}%
\pgfpathlineto{\pgfqpoint{5.400000in}{3.600000in}}%
\pgfusepath{stroke}%
\end{pgfscope}%
\begin{pgfscope}%
\pgfsetrectcap%
\pgfsetmiterjoin%
\pgfsetlinewidth{1.003750pt}%
\definecolor{currentstroke}{rgb}{0.000000,0.000000,0.000000}%
\pgfsetstrokecolor{currentstroke}%
\pgfsetdash{}{0pt}%
\pgfpathmoveto{\pgfqpoint{5.400000in}{0.500000in}}%
\pgfpathlineto{\pgfqpoint{5.400000in}{3.600000in}}%
\pgfusepath{stroke}%
\end{pgfscope}%
\begin{pgfscope}%
\pgfsetrectcap%
\pgfsetmiterjoin%
\pgfsetlinewidth{1.003750pt}%
\definecolor{currentstroke}{rgb}{0.000000,0.000000,0.000000}%
\pgfsetstrokecolor{currentstroke}%
\pgfsetdash{}{0pt}%
\pgfpathmoveto{\pgfqpoint{0.750000in}{0.500000in}}%
\pgfpathlineto{\pgfqpoint{5.400000in}{0.500000in}}%
\pgfusepath{stroke}%
\end{pgfscope}%
\begin{pgfscope}%
\pgfsetrectcap%
\pgfsetmiterjoin%
\pgfsetlinewidth{1.003750pt}%
\definecolor{currentstroke}{rgb}{0.000000,0.000000,0.000000}%
\pgfsetstrokecolor{currentstroke}%
\pgfsetdash{}{0pt}%
\pgfpathmoveto{\pgfqpoint{0.750000in}{0.500000in}}%
\pgfpathlineto{\pgfqpoint{0.750000in}{3.600000in}}%
\pgfusepath{stroke}%
\end{pgfscope}%
\begin{pgfscope}%
\pgfpathrectangle{\pgfqpoint{0.750000in}{0.500000in}}{\pgfqpoint{4.650000in}{3.100000in}} %
\pgfusepath{clip}%
\pgfsetrectcap%
\pgfsetroundjoin%
\pgfsetlinewidth{0.501875pt}%
\definecolor{currentstroke}{rgb}{0.000000,0.000000,0.000000}%
\pgfsetstrokecolor{currentstroke}%
\pgfsetstrokeopacity{0.100000}%
\pgfsetdash{}{0pt}%
\pgfpathmoveto{\pgfqpoint{0.750000in}{0.500000in}}%
\pgfpathlineto{\pgfqpoint{0.750000in}{3.600000in}}%
\pgfusepath{stroke}%
\end{pgfscope}%
\begin{pgfscope}%
\pgfsetbuttcap%
\pgfsetroundjoin%
\definecolor{currentfill}{rgb}{0.000000,0.000000,0.000000}%
\pgfsetfillcolor{currentfill}%
\pgfsetlinewidth{0.501875pt}%
\definecolor{currentstroke}{rgb}{0.000000,0.000000,0.000000}%
\pgfsetstrokecolor{currentstroke}%
\pgfsetdash{}{0pt}%
\pgfsys@defobject{currentmarker}{\pgfqpoint{0.000000in}{0.000000in}}{\pgfqpoint{0.000000in}{0.055556in}}{%
\pgfpathmoveto{\pgfqpoint{0.000000in}{0.000000in}}%
\pgfpathlineto{\pgfqpoint{0.000000in}{0.055556in}}%
\pgfusepath{stroke,fill}%
}%
\begin{pgfscope}%
\pgfsys@transformshift{0.750000in}{0.500000in}%
\pgfsys@useobject{currentmarker}{}%
\end{pgfscope}%
\end{pgfscope}%
\begin{pgfscope}%
\pgfsetbuttcap%
\pgfsetroundjoin%
\definecolor{currentfill}{rgb}{0.000000,0.000000,0.000000}%
\pgfsetfillcolor{currentfill}%
\pgfsetlinewidth{0.501875pt}%
\definecolor{currentstroke}{rgb}{0.000000,0.000000,0.000000}%
\pgfsetstrokecolor{currentstroke}%
\pgfsetdash{}{0pt}%
\pgfsys@defobject{currentmarker}{\pgfqpoint{0.000000in}{-0.055556in}}{\pgfqpoint{0.000000in}{0.000000in}}{%
\pgfpathmoveto{\pgfqpoint{0.000000in}{0.000000in}}%
\pgfpathlineto{\pgfqpoint{0.000000in}{-0.055556in}}%
\pgfusepath{stroke,fill}%
}%
\begin{pgfscope}%
\pgfsys@transformshift{0.750000in}{3.600000in}%
\pgfsys@useobject{currentmarker}{}%
\end{pgfscope}%
\end{pgfscope}%
\begin{pgfscope}%
\pgftext[x=0.750000in,y=0.444444in,,top]{\rmfamily\fontsize{10.000000}{12.000000}\selectfont 0}%
\end{pgfscope}%
\begin{pgfscope}%
\pgfpathrectangle{\pgfqpoint{0.750000in}{0.500000in}}{\pgfqpoint{4.650000in}{3.100000in}} %
\pgfusepath{clip}%
\pgfsetrectcap%
\pgfsetroundjoin%
\pgfsetlinewidth{0.501875pt}%
\definecolor{currentstroke}{rgb}{0.000000,0.000000,0.000000}%
\pgfsetstrokecolor{currentstroke}%
\pgfsetstrokeopacity{0.100000}%
\pgfsetdash{}{0pt}%
\pgfpathmoveto{\pgfqpoint{1.680000in}{0.500000in}}%
\pgfpathlineto{\pgfqpoint{1.680000in}{3.600000in}}%
\pgfusepath{stroke}%
\end{pgfscope}%
\begin{pgfscope}%
\pgfsetbuttcap%
\pgfsetroundjoin%
\definecolor{currentfill}{rgb}{0.000000,0.000000,0.000000}%
\pgfsetfillcolor{currentfill}%
\pgfsetlinewidth{0.501875pt}%
\definecolor{currentstroke}{rgb}{0.000000,0.000000,0.000000}%
\pgfsetstrokecolor{currentstroke}%
\pgfsetdash{}{0pt}%
\pgfsys@defobject{currentmarker}{\pgfqpoint{0.000000in}{0.000000in}}{\pgfqpoint{0.000000in}{0.055556in}}{%
\pgfpathmoveto{\pgfqpoint{0.000000in}{0.000000in}}%
\pgfpathlineto{\pgfqpoint{0.000000in}{0.055556in}}%
\pgfusepath{stroke,fill}%
}%
\begin{pgfscope}%
\pgfsys@transformshift{1.680000in}{0.500000in}%
\pgfsys@useobject{currentmarker}{}%
\end{pgfscope}%
\end{pgfscope}%
\begin{pgfscope}%
\pgfsetbuttcap%
\pgfsetroundjoin%
\definecolor{currentfill}{rgb}{0.000000,0.000000,0.000000}%
\pgfsetfillcolor{currentfill}%
\pgfsetlinewidth{0.501875pt}%
\definecolor{currentstroke}{rgb}{0.000000,0.000000,0.000000}%
\pgfsetstrokecolor{currentstroke}%
\pgfsetdash{}{0pt}%
\pgfsys@defobject{currentmarker}{\pgfqpoint{0.000000in}{-0.055556in}}{\pgfqpoint{0.000000in}{0.000000in}}{%
\pgfpathmoveto{\pgfqpoint{0.000000in}{0.000000in}}%
\pgfpathlineto{\pgfqpoint{0.000000in}{-0.055556in}}%
\pgfusepath{stroke,fill}%
}%
\begin{pgfscope}%
\pgfsys@transformshift{1.680000in}{3.600000in}%
\pgfsys@useobject{currentmarker}{}%
\end{pgfscope}%
\end{pgfscope}%
\begin{pgfscope}%
\pgftext[x=1.680000in,y=0.444444in,,top]{\rmfamily\fontsize{10.000000}{12.000000}\selectfont 20}%
\end{pgfscope}%
\begin{pgfscope}%
\pgfpathrectangle{\pgfqpoint{0.750000in}{0.500000in}}{\pgfqpoint{4.650000in}{3.100000in}} %
\pgfusepath{clip}%
\pgfsetrectcap%
\pgfsetroundjoin%
\pgfsetlinewidth{0.501875pt}%
\definecolor{currentstroke}{rgb}{0.000000,0.000000,0.000000}%
\pgfsetstrokecolor{currentstroke}%
\pgfsetstrokeopacity{0.100000}%
\pgfsetdash{}{0pt}%
\pgfpathmoveto{\pgfqpoint{2.610000in}{0.500000in}}%
\pgfpathlineto{\pgfqpoint{2.610000in}{3.600000in}}%
\pgfusepath{stroke}%
\end{pgfscope}%
\begin{pgfscope}%
\pgfsetbuttcap%
\pgfsetroundjoin%
\definecolor{currentfill}{rgb}{0.000000,0.000000,0.000000}%
\pgfsetfillcolor{currentfill}%
\pgfsetlinewidth{0.501875pt}%
\definecolor{currentstroke}{rgb}{0.000000,0.000000,0.000000}%
\pgfsetstrokecolor{currentstroke}%
\pgfsetdash{}{0pt}%
\pgfsys@defobject{currentmarker}{\pgfqpoint{0.000000in}{0.000000in}}{\pgfqpoint{0.000000in}{0.055556in}}{%
\pgfpathmoveto{\pgfqpoint{0.000000in}{0.000000in}}%
\pgfpathlineto{\pgfqpoint{0.000000in}{0.055556in}}%
\pgfusepath{stroke,fill}%
}%
\begin{pgfscope}%
\pgfsys@transformshift{2.610000in}{0.500000in}%
\pgfsys@useobject{currentmarker}{}%
\end{pgfscope}%
\end{pgfscope}%
\begin{pgfscope}%
\pgfsetbuttcap%
\pgfsetroundjoin%
\definecolor{currentfill}{rgb}{0.000000,0.000000,0.000000}%
\pgfsetfillcolor{currentfill}%
\pgfsetlinewidth{0.501875pt}%
\definecolor{currentstroke}{rgb}{0.000000,0.000000,0.000000}%
\pgfsetstrokecolor{currentstroke}%
\pgfsetdash{}{0pt}%
\pgfsys@defobject{currentmarker}{\pgfqpoint{0.000000in}{-0.055556in}}{\pgfqpoint{0.000000in}{0.000000in}}{%
\pgfpathmoveto{\pgfqpoint{0.000000in}{0.000000in}}%
\pgfpathlineto{\pgfqpoint{0.000000in}{-0.055556in}}%
\pgfusepath{stroke,fill}%
}%
\begin{pgfscope}%
\pgfsys@transformshift{2.610000in}{3.600000in}%
\pgfsys@useobject{currentmarker}{}%
\end{pgfscope}%
\end{pgfscope}%
\begin{pgfscope}%
\pgftext[x=2.610000in,y=0.444444in,,top]{\rmfamily\fontsize{10.000000}{12.000000}\selectfont 40}%
\end{pgfscope}%
\begin{pgfscope}%
\pgfpathrectangle{\pgfqpoint{0.750000in}{0.500000in}}{\pgfqpoint{4.650000in}{3.100000in}} %
\pgfusepath{clip}%
\pgfsetrectcap%
\pgfsetroundjoin%
\pgfsetlinewidth{0.501875pt}%
\definecolor{currentstroke}{rgb}{0.000000,0.000000,0.000000}%
\pgfsetstrokecolor{currentstroke}%
\pgfsetstrokeopacity{0.100000}%
\pgfsetdash{}{0pt}%
\pgfpathmoveto{\pgfqpoint{3.540000in}{0.500000in}}%
\pgfpathlineto{\pgfqpoint{3.540000in}{3.600000in}}%
\pgfusepath{stroke}%
\end{pgfscope}%
\begin{pgfscope}%
\pgfsetbuttcap%
\pgfsetroundjoin%
\definecolor{currentfill}{rgb}{0.000000,0.000000,0.000000}%
\pgfsetfillcolor{currentfill}%
\pgfsetlinewidth{0.501875pt}%
\definecolor{currentstroke}{rgb}{0.000000,0.000000,0.000000}%
\pgfsetstrokecolor{currentstroke}%
\pgfsetdash{}{0pt}%
\pgfsys@defobject{currentmarker}{\pgfqpoint{0.000000in}{0.000000in}}{\pgfqpoint{0.000000in}{0.055556in}}{%
\pgfpathmoveto{\pgfqpoint{0.000000in}{0.000000in}}%
\pgfpathlineto{\pgfqpoint{0.000000in}{0.055556in}}%
\pgfusepath{stroke,fill}%
}%
\begin{pgfscope}%
\pgfsys@transformshift{3.540000in}{0.500000in}%
\pgfsys@useobject{currentmarker}{}%
\end{pgfscope}%
\end{pgfscope}%
\begin{pgfscope}%
\pgfsetbuttcap%
\pgfsetroundjoin%
\definecolor{currentfill}{rgb}{0.000000,0.000000,0.000000}%
\pgfsetfillcolor{currentfill}%
\pgfsetlinewidth{0.501875pt}%
\definecolor{currentstroke}{rgb}{0.000000,0.000000,0.000000}%
\pgfsetstrokecolor{currentstroke}%
\pgfsetdash{}{0pt}%
\pgfsys@defobject{currentmarker}{\pgfqpoint{0.000000in}{-0.055556in}}{\pgfqpoint{0.000000in}{0.000000in}}{%
\pgfpathmoveto{\pgfqpoint{0.000000in}{0.000000in}}%
\pgfpathlineto{\pgfqpoint{0.000000in}{-0.055556in}}%
\pgfusepath{stroke,fill}%
}%
\begin{pgfscope}%
\pgfsys@transformshift{3.540000in}{3.600000in}%
\pgfsys@useobject{currentmarker}{}%
\end{pgfscope}%
\end{pgfscope}%
\begin{pgfscope}%
\pgftext[x=3.540000in,y=0.444444in,,top]{\rmfamily\fontsize{10.000000}{12.000000}\selectfont 60}%
\end{pgfscope}%
\begin{pgfscope}%
\pgfpathrectangle{\pgfqpoint{0.750000in}{0.500000in}}{\pgfqpoint{4.650000in}{3.100000in}} %
\pgfusepath{clip}%
\pgfsetrectcap%
\pgfsetroundjoin%
\pgfsetlinewidth{0.501875pt}%
\definecolor{currentstroke}{rgb}{0.000000,0.000000,0.000000}%
\pgfsetstrokecolor{currentstroke}%
\pgfsetstrokeopacity{0.100000}%
\pgfsetdash{}{0pt}%
\pgfpathmoveto{\pgfqpoint{4.470000in}{0.500000in}}%
\pgfpathlineto{\pgfqpoint{4.470000in}{3.600000in}}%
\pgfusepath{stroke}%
\end{pgfscope}%
\begin{pgfscope}%
\pgfsetbuttcap%
\pgfsetroundjoin%
\definecolor{currentfill}{rgb}{0.000000,0.000000,0.000000}%
\pgfsetfillcolor{currentfill}%
\pgfsetlinewidth{0.501875pt}%
\definecolor{currentstroke}{rgb}{0.000000,0.000000,0.000000}%
\pgfsetstrokecolor{currentstroke}%
\pgfsetdash{}{0pt}%
\pgfsys@defobject{currentmarker}{\pgfqpoint{0.000000in}{0.000000in}}{\pgfqpoint{0.000000in}{0.055556in}}{%
\pgfpathmoveto{\pgfqpoint{0.000000in}{0.000000in}}%
\pgfpathlineto{\pgfqpoint{0.000000in}{0.055556in}}%
\pgfusepath{stroke,fill}%
}%
\begin{pgfscope}%
\pgfsys@transformshift{4.470000in}{0.500000in}%
\pgfsys@useobject{currentmarker}{}%
\end{pgfscope}%
\end{pgfscope}%
\begin{pgfscope}%
\pgfsetbuttcap%
\pgfsetroundjoin%
\definecolor{currentfill}{rgb}{0.000000,0.000000,0.000000}%
\pgfsetfillcolor{currentfill}%
\pgfsetlinewidth{0.501875pt}%
\definecolor{currentstroke}{rgb}{0.000000,0.000000,0.000000}%
\pgfsetstrokecolor{currentstroke}%
\pgfsetdash{}{0pt}%
\pgfsys@defobject{currentmarker}{\pgfqpoint{0.000000in}{-0.055556in}}{\pgfqpoint{0.000000in}{0.000000in}}{%
\pgfpathmoveto{\pgfqpoint{0.000000in}{0.000000in}}%
\pgfpathlineto{\pgfqpoint{0.000000in}{-0.055556in}}%
\pgfusepath{stroke,fill}%
}%
\begin{pgfscope}%
\pgfsys@transformshift{4.470000in}{3.600000in}%
\pgfsys@useobject{currentmarker}{}%
\end{pgfscope}%
\end{pgfscope}%
\begin{pgfscope}%
\pgftext[x=4.470000in,y=0.444444in,,top]{\rmfamily\fontsize{10.000000}{12.000000}\selectfont 80}%
\end{pgfscope}%
\begin{pgfscope}%
\pgfpathrectangle{\pgfqpoint{0.750000in}{0.500000in}}{\pgfqpoint{4.650000in}{3.100000in}} %
\pgfusepath{clip}%
\pgfsetrectcap%
\pgfsetroundjoin%
\pgfsetlinewidth{0.501875pt}%
\definecolor{currentstroke}{rgb}{0.000000,0.000000,0.000000}%
\pgfsetstrokecolor{currentstroke}%
\pgfsetstrokeopacity{0.100000}%
\pgfsetdash{}{0pt}%
\pgfpathmoveto{\pgfqpoint{5.400000in}{0.500000in}}%
\pgfpathlineto{\pgfqpoint{5.400000in}{3.600000in}}%
\pgfusepath{stroke}%
\end{pgfscope}%
\begin{pgfscope}%
\pgfsetbuttcap%
\pgfsetroundjoin%
\definecolor{currentfill}{rgb}{0.000000,0.000000,0.000000}%
\pgfsetfillcolor{currentfill}%
\pgfsetlinewidth{0.501875pt}%
\definecolor{currentstroke}{rgb}{0.000000,0.000000,0.000000}%
\pgfsetstrokecolor{currentstroke}%
\pgfsetdash{}{0pt}%
\pgfsys@defobject{currentmarker}{\pgfqpoint{0.000000in}{0.000000in}}{\pgfqpoint{0.000000in}{0.055556in}}{%
\pgfpathmoveto{\pgfqpoint{0.000000in}{0.000000in}}%
\pgfpathlineto{\pgfqpoint{0.000000in}{0.055556in}}%
\pgfusepath{stroke,fill}%
}%
\begin{pgfscope}%
\pgfsys@transformshift{5.400000in}{0.500000in}%
\pgfsys@useobject{currentmarker}{}%
\end{pgfscope}%
\end{pgfscope}%
\begin{pgfscope}%
\pgfsetbuttcap%
\pgfsetroundjoin%
\definecolor{currentfill}{rgb}{0.000000,0.000000,0.000000}%
\pgfsetfillcolor{currentfill}%
\pgfsetlinewidth{0.501875pt}%
\definecolor{currentstroke}{rgb}{0.000000,0.000000,0.000000}%
\pgfsetstrokecolor{currentstroke}%
\pgfsetdash{}{0pt}%
\pgfsys@defobject{currentmarker}{\pgfqpoint{0.000000in}{-0.055556in}}{\pgfqpoint{0.000000in}{0.000000in}}{%
\pgfpathmoveto{\pgfqpoint{0.000000in}{0.000000in}}%
\pgfpathlineto{\pgfqpoint{0.000000in}{-0.055556in}}%
\pgfusepath{stroke,fill}%
}%
\begin{pgfscope}%
\pgfsys@transformshift{5.400000in}{3.600000in}%
\pgfsys@useobject{currentmarker}{}%
\end{pgfscope}%
\end{pgfscope}%
\begin{pgfscope}%
\pgftext[x=5.400000in,y=0.444444in,,top]{\rmfamily\fontsize{10.000000}{12.000000}\selectfont 100}%
\end{pgfscope}%
\begin{pgfscope}%
\pgftext[x=3.075000in,y=0.251667in,,top]{\rmfamily\fontsize{10.000000}{12.000000}\selectfont Seed set size}%
\end{pgfscope}%
\begin{pgfscope}%
\pgfpathrectangle{\pgfqpoint{0.750000in}{0.500000in}}{\pgfqpoint{4.650000in}{3.100000in}} %
\pgfusepath{clip}%
\pgfsetrectcap%
\pgfsetroundjoin%
\pgfsetlinewidth{0.501875pt}%
\definecolor{currentstroke}{rgb}{0.000000,0.000000,0.000000}%
\pgfsetstrokecolor{currentstroke}%
\pgfsetstrokeopacity{0.100000}%
\pgfsetdash{}{0pt}%
\pgfpathmoveto{\pgfqpoint{0.750000in}{0.500000in}}%
\pgfpathlineto{\pgfqpoint{5.400000in}{0.500000in}}%
\pgfusepath{stroke}%
\end{pgfscope}%
\begin{pgfscope}%
\pgfsetbuttcap%
\pgfsetroundjoin%
\definecolor{currentfill}{rgb}{0.000000,0.000000,0.000000}%
\pgfsetfillcolor{currentfill}%
\pgfsetlinewidth{0.501875pt}%
\definecolor{currentstroke}{rgb}{0.000000,0.000000,0.000000}%
\pgfsetstrokecolor{currentstroke}%
\pgfsetdash{}{0pt}%
\pgfsys@defobject{currentmarker}{\pgfqpoint{0.000000in}{0.000000in}}{\pgfqpoint{0.055556in}{0.000000in}}{%
\pgfpathmoveto{\pgfqpoint{0.000000in}{0.000000in}}%
\pgfpathlineto{\pgfqpoint{0.055556in}{0.000000in}}%
\pgfusepath{stroke,fill}%
}%
\begin{pgfscope}%
\pgfsys@transformshift{0.750000in}{0.500000in}%
\pgfsys@useobject{currentmarker}{}%
\end{pgfscope}%
\end{pgfscope}%
\begin{pgfscope}%
\pgfsetbuttcap%
\pgfsetroundjoin%
\definecolor{currentfill}{rgb}{0.000000,0.000000,0.000000}%
\pgfsetfillcolor{currentfill}%
\pgfsetlinewidth{0.501875pt}%
\definecolor{currentstroke}{rgb}{0.000000,0.000000,0.000000}%
\pgfsetstrokecolor{currentstroke}%
\pgfsetdash{}{0pt}%
\pgfsys@defobject{currentmarker}{\pgfqpoint{-0.055556in}{0.000000in}}{\pgfqpoint{0.000000in}{0.000000in}}{%
\pgfpathmoveto{\pgfqpoint{0.000000in}{0.000000in}}%
\pgfpathlineto{\pgfqpoint{-0.055556in}{0.000000in}}%
\pgfusepath{stroke,fill}%
}%
\begin{pgfscope}%
\pgfsys@transformshift{5.400000in}{0.500000in}%
\pgfsys@useobject{currentmarker}{}%
\end{pgfscope}%
\end{pgfscope}%
\begin{pgfscope}%
\pgftext[x=0.694444in,y=0.500000in,right,]{\rmfamily\fontsize{10.000000}{12.000000}\selectfont 0}%
\end{pgfscope}%
\begin{pgfscope}%
\pgfpathrectangle{\pgfqpoint{0.750000in}{0.500000in}}{\pgfqpoint{4.650000in}{3.100000in}} %
\pgfusepath{clip}%
\pgfsetrectcap%
\pgfsetroundjoin%
\pgfsetlinewidth{0.501875pt}%
\definecolor{currentstroke}{rgb}{0.000000,0.000000,0.000000}%
\pgfsetstrokecolor{currentstroke}%
\pgfsetstrokeopacity{0.100000}%
\pgfsetdash{}{0pt}%
\pgfpathmoveto{\pgfqpoint{0.750000in}{1.120000in}}%
\pgfpathlineto{\pgfqpoint{5.400000in}{1.120000in}}%
\pgfusepath{stroke}%
\end{pgfscope}%
\begin{pgfscope}%
\pgfsetbuttcap%
\pgfsetroundjoin%
\definecolor{currentfill}{rgb}{0.000000,0.000000,0.000000}%
\pgfsetfillcolor{currentfill}%
\pgfsetlinewidth{0.501875pt}%
\definecolor{currentstroke}{rgb}{0.000000,0.000000,0.000000}%
\pgfsetstrokecolor{currentstroke}%
\pgfsetdash{}{0pt}%
\pgfsys@defobject{currentmarker}{\pgfqpoint{0.000000in}{0.000000in}}{\pgfqpoint{0.055556in}{0.000000in}}{%
\pgfpathmoveto{\pgfqpoint{0.000000in}{0.000000in}}%
\pgfpathlineto{\pgfqpoint{0.055556in}{0.000000in}}%
\pgfusepath{stroke,fill}%
}%
\begin{pgfscope}%
\pgfsys@transformshift{0.750000in}{1.120000in}%
\pgfsys@useobject{currentmarker}{}%
\end{pgfscope}%
\end{pgfscope}%
\begin{pgfscope}%
\pgfsetbuttcap%
\pgfsetroundjoin%
\definecolor{currentfill}{rgb}{0.000000,0.000000,0.000000}%
\pgfsetfillcolor{currentfill}%
\pgfsetlinewidth{0.501875pt}%
\definecolor{currentstroke}{rgb}{0.000000,0.000000,0.000000}%
\pgfsetstrokecolor{currentstroke}%
\pgfsetdash{}{0pt}%
\pgfsys@defobject{currentmarker}{\pgfqpoint{-0.055556in}{0.000000in}}{\pgfqpoint{0.000000in}{0.000000in}}{%
\pgfpathmoveto{\pgfqpoint{0.000000in}{0.000000in}}%
\pgfpathlineto{\pgfqpoint{-0.055556in}{0.000000in}}%
\pgfusepath{stroke,fill}%
}%
\begin{pgfscope}%
\pgfsys@transformshift{5.400000in}{1.120000in}%
\pgfsys@useobject{currentmarker}{}%
\end{pgfscope}%
\end{pgfscope}%
\begin{pgfscope}%
\pgftext[x=0.694444in,y=1.120000in,right,]{\rmfamily\fontsize{10.000000}{12.000000}\selectfont 20}%
\end{pgfscope}%
\begin{pgfscope}%
\pgfpathrectangle{\pgfqpoint{0.750000in}{0.500000in}}{\pgfqpoint{4.650000in}{3.100000in}} %
\pgfusepath{clip}%
\pgfsetrectcap%
\pgfsetroundjoin%
\pgfsetlinewidth{0.501875pt}%
\definecolor{currentstroke}{rgb}{0.000000,0.000000,0.000000}%
\pgfsetstrokecolor{currentstroke}%
\pgfsetstrokeopacity{0.100000}%
\pgfsetdash{}{0pt}%
\pgfpathmoveto{\pgfqpoint{0.750000in}{1.740000in}}%
\pgfpathlineto{\pgfqpoint{5.400000in}{1.740000in}}%
\pgfusepath{stroke}%
\end{pgfscope}%
\begin{pgfscope}%
\pgfsetbuttcap%
\pgfsetroundjoin%
\definecolor{currentfill}{rgb}{0.000000,0.000000,0.000000}%
\pgfsetfillcolor{currentfill}%
\pgfsetlinewidth{0.501875pt}%
\definecolor{currentstroke}{rgb}{0.000000,0.000000,0.000000}%
\pgfsetstrokecolor{currentstroke}%
\pgfsetdash{}{0pt}%
\pgfsys@defobject{currentmarker}{\pgfqpoint{0.000000in}{0.000000in}}{\pgfqpoint{0.055556in}{0.000000in}}{%
\pgfpathmoveto{\pgfqpoint{0.000000in}{0.000000in}}%
\pgfpathlineto{\pgfqpoint{0.055556in}{0.000000in}}%
\pgfusepath{stroke,fill}%
}%
\begin{pgfscope}%
\pgfsys@transformshift{0.750000in}{1.740000in}%
\pgfsys@useobject{currentmarker}{}%
\end{pgfscope}%
\end{pgfscope}%
\begin{pgfscope}%
\pgfsetbuttcap%
\pgfsetroundjoin%
\definecolor{currentfill}{rgb}{0.000000,0.000000,0.000000}%
\pgfsetfillcolor{currentfill}%
\pgfsetlinewidth{0.501875pt}%
\definecolor{currentstroke}{rgb}{0.000000,0.000000,0.000000}%
\pgfsetstrokecolor{currentstroke}%
\pgfsetdash{}{0pt}%
\pgfsys@defobject{currentmarker}{\pgfqpoint{-0.055556in}{0.000000in}}{\pgfqpoint{0.000000in}{0.000000in}}{%
\pgfpathmoveto{\pgfqpoint{0.000000in}{0.000000in}}%
\pgfpathlineto{\pgfqpoint{-0.055556in}{0.000000in}}%
\pgfusepath{stroke,fill}%
}%
\begin{pgfscope}%
\pgfsys@transformshift{5.400000in}{1.740000in}%
\pgfsys@useobject{currentmarker}{}%
\end{pgfscope}%
\end{pgfscope}%
\begin{pgfscope}%
\pgftext[x=0.694444in,y=1.740000in,right,]{\rmfamily\fontsize{10.000000}{12.000000}\selectfont 40}%
\end{pgfscope}%
\begin{pgfscope}%
\pgfpathrectangle{\pgfqpoint{0.750000in}{0.500000in}}{\pgfqpoint{4.650000in}{3.100000in}} %
\pgfusepath{clip}%
\pgfsetrectcap%
\pgfsetroundjoin%
\pgfsetlinewidth{0.501875pt}%
\definecolor{currentstroke}{rgb}{0.000000,0.000000,0.000000}%
\pgfsetstrokecolor{currentstroke}%
\pgfsetstrokeopacity{0.100000}%
\pgfsetdash{}{0pt}%
\pgfpathmoveto{\pgfqpoint{0.750000in}{2.360000in}}%
\pgfpathlineto{\pgfqpoint{5.400000in}{2.360000in}}%
\pgfusepath{stroke}%
\end{pgfscope}%
\begin{pgfscope}%
\pgfsetbuttcap%
\pgfsetroundjoin%
\definecolor{currentfill}{rgb}{0.000000,0.000000,0.000000}%
\pgfsetfillcolor{currentfill}%
\pgfsetlinewidth{0.501875pt}%
\definecolor{currentstroke}{rgb}{0.000000,0.000000,0.000000}%
\pgfsetstrokecolor{currentstroke}%
\pgfsetdash{}{0pt}%
\pgfsys@defobject{currentmarker}{\pgfqpoint{0.000000in}{0.000000in}}{\pgfqpoint{0.055556in}{0.000000in}}{%
\pgfpathmoveto{\pgfqpoint{0.000000in}{0.000000in}}%
\pgfpathlineto{\pgfqpoint{0.055556in}{0.000000in}}%
\pgfusepath{stroke,fill}%
}%
\begin{pgfscope}%
\pgfsys@transformshift{0.750000in}{2.360000in}%
\pgfsys@useobject{currentmarker}{}%
\end{pgfscope}%
\end{pgfscope}%
\begin{pgfscope}%
\pgfsetbuttcap%
\pgfsetroundjoin%
\definecolor{currentfill}{rgb}{0.000000,0.000000,0.000000}%
\pgfsetfillcolor{currentfill}%
\pgfsetlinewidth{0.501875pt}%
\definecolor{currentstroke}{rgb}{0.000000,0.000000,0.000000}%
\pgfsetstrokecolor{currentstroke}%
\pgfsetdash{}{0pt}%
\pgfsys@defobject{currentmarker}{\pgfqpoint{-0.055556in}{0.000000in}}{\pgfqpoint{0.000000in}{0.000000in}}{%
\pgfpathmoveto{\pgfqpoint{0.000000in}{0.000000in}}%
\pgfpathlineto{\pgfqpoint{-0.055556in}{0.000000in}}%
\pgfusepath{stroke,fill}%
}%
\begin{pgfscope}%
\pgfsys@transformshift{5.400000in}{2.360000in}%
\pgfsys@useobject{currentmarker}{}%
\end{pgfscope}%
\end{pgfscope}%
\begin{pgfscope}%
\pgftext[x=0.694444in,y=2.360000in,right,]{\rmfamily\fontsize{10.000000}{12.000000}\selectfont 60}%
\end{pgfscope}%
\begin{pgfscope}%
\pgfpathrectangle{\pgfqpoint{0.750000in}{0.500000in}}{\pgfqpoint{4.650000in}{3.100000in}} %
\pgfusepath{clip}%
\pgfsetrectcap%
\pgfsetroundjoin%
\pgfsetlinewidth{0.501875pt}%
\definecolor{currentstroke}{rgb}{0.000000,0.000000,0.000000}%
\pgfsetstrokecolor{currentstroke}%
\pgfsetstrokeopacity{0.100000}%
\pgfsetdash{}{0pt}%
\pgfpathmoveto{\pgfqpoint{0.750000in}{2.980000in}}%
\pgfpathlineto{\pgfqpoint{5.400000in}{2.980000in}}%
\pgfusepath{stroke}%
\end{pgfscope}%
\begin{pgfscope}%
\pgfsetbuttcap%
\pgfsetroundjoin%
\definecolor{currentfill}{rgb}{0.000000,0.000000,0.000000}%
\pgfsetfillcolor{currentfill}%
\pgfsetlinewidth{0.501875pt}%
\definecolor{currentstroke}{rgb}{0.000000,0.000000,0.000000}%
\pgfsetstrokecolor{currentstroke}%
\pgfsetdash{}{0pt}%
\pgfsys@defobject{currentmarker}{\pgfqpoint{0.000000in}{0.000000in}}{\pgfqpoint{0.055556in}{0.000000in}}{%
\pgfpathmoveto{\pgfqpoint{0.000000in}{0.000000in}}%
\pgfpathlineto{\pgfqpoint{0.055556in}{0.000000in}}%
\pgfusepath{stroke,fill}%
}%
\begin{pgfscope}%
\pgfsys@transformshift{0.750000in}{2.980000in}%
\pgfsys@useobject{currentmarker}{}%
\end{pgfscope}%
\end{pgfscope}%
\begin{pgfscope}%
\pgfsetbuttcap%
\pgfsetroundjoin%
\definecolor{currentfill}{rgb}{0.000000,0.000000,0.000000}%
\pgfsetfillcolor{currentfill}%
\pgfsetlinewidth{0.501875pt}%
\definecolor{currentstroke}{rgb}{0.000000,0.000000,0.000000}%
\pgfsetstrokecolor{currentstroke}%
\pgfsetdash{}{0pt}%
\pgfsys@defobject{currentmarker}{\pgfqpoint{-0.055556in}{0.000000in}}{\pgfqpoint{0.000000in}{0.000000in}}{%
\pgfpathmoveto{\pgfqpoint{0.000000in}{0.000000in}}%
\pgfpathlineto{\pgfqpoint{-0.055556in}{0.000000in}}%
\pgfusepath{stroke,fill}%
}%
\begin{pgfscope}%
\pgfsys@transformshift{5.400000in}{2.980000in}%
\pgfsys@useobject{currentmarker}{}%
\end{pgfscope}%
\end{pgfscope}%
\begin{pgfscope}%
\pgftext[x=0.694444in,y=2.980000in,right,]{\rmfamily\fontsize{10.000000}{12.000000}\selectfont 80}%
\end{pgfscope}%
\begin{pgfscope}%
\pgfpathrectangle{\pgfqpoint{0.750000in}{0.500000in}}{\pgfqpoint{4.650000in}{3.100000in}} %
\pgfusepath{clip}%
\pgfsetrectcap%
\pgfsetroundjoin%
\pgfsetlinewidth{0.501875pt}%
\definecolor{currentstroke}{rgb}{0.000000,0.000000,0.000000}%
\pgfsetstrokecolor{currentstroke}%
\pgfsetstrokeopacity{0.100000}%
\pgfsetdash{}{0pt}%
\pgfpathmoveto{\pgfqpoint{0.750000in}{3.600000in}}%
\pgfpathlineto{\pgfqpoint{5.400000in}{3.600000in}}%
\pgfusepath{stroke}%
\end{pgfscope}%
\begin{pgfscope}%
\pgfsetbuttcap%
\pgfsetroundjoin%
\definecolor{currentfill}{rgb}{0.000000,0.000000,0.000000}%
\pgfsetfillcolor{currentfill}%
\pgfsetlinewidth{0.501875pt}%
\definecolor{currentstroke}{rgb}{0.000000,0.000000,0.000000}%
\pgfsetstrokecolor{currentstroke}%
\pgfsetdash{}{0pt}%
\pgfsys@defobject{currentmarker}{\pgfqpoint{0.000000in}{0.000000in}}{\pgfqpoint{0.055556in}{0.000000in}}{%
\pgfpathmoveto{\pgfqpoint{0.000000in}{0.000000in}}%
\pgfpathlineto{\pgfqpoint{0.055556in}{0.000000in}}%
\pgfusepath{stroke,fill}%
}%
\begin{pgfscope}%
\pgfsys@transformshift{0.750000in}{3.600000in}%
\pgfsys@useobject{currentmarker}{}%
\end{pgfscope}%
\end{pgfscope}%
\begin{pgfscope}%
\pgfsetbuttcap%
\pgfsetroundjoin%
\definecolor{currentfill}{rgb}{0.000000,0.000000,0.000000}%
\pgfsetfillcolor{currentfill}%
\pgfsetlinewidth{0.501875pt}%
\definecolor{currentstroke}{rgb}{0.000000,0.000000,0.000000}%
\pgfsetstrokecolor{currentstroke}%
\pgfsetdash{}{0pt}%
\pgfsys@defobject{currentmarker}{\pgfqpoint{-0.055556in}{0.000000in}}{\pgfqpoint{0.000000in}{0.000000in}}{%
\pgfpathmoveto{\pgfqpoint{0.000000in}{0.000000in}}%
\pgfpathlineto{\pgfqpoint{-0.055556in}{0.000000in}}%
\pgfusepath{stroke,fill}%
}%
\begin{pgfscope}%
\pgfsys@transformshift{5.400000in}{3.600000in}%
\pgfsys@useobject{currentmarker}{}%
\end{pgfscope}%
\end{pgfscope}%
\begin{pgfscope}%
\pgftext[x=0.694444in,y=3.600000in,right,]{\rmfamily\fontsize{10.000000}{12.000000}\selectfont 100}%
\end{pgfscope}%
\begin{pgfscope}%
\pgftext[x=0.416667in,y=2.050000in,,bottom,rotate=90.000000]{\rmfamily\fontsize{10.000000}{12.000000}\selectfont Optimal rank}%
\end{pgfscope}%
\begin{pgfscope}%
\pgfsetbuttcap%
\pgfsetmiterjoin%
\pgfsetlinewidth{1.003750pt}%
\definecolor{currentstroke}{rgb}{0.000000,0.000000,0.000000}%
\pgfsetstrokecolor{currentstroke}%
\pgfsetdash{}{0pt}%
\pgfpathmoveto{\pgfqpoint{0.833333in}{2.303667in}}%
\pgfpathlineto{\pgfqpoint{3.026111in}{2.303667in}}%
\pgfpathlineto{\pgfqpoint{3.026111in}{3.516667in}}%
\pgfpathlineto{\pgfqpoint{0.833333in}{3.516667in}}%
\pgfpathclose%
\pgfusepath{stroke}%
\end{pgfscope}%
\begin{pgfscope}%
\pgfsetrectcap%
\pgfsetroundjoin%
\pgfsetlinewidth{1.003750pt}%
\definecolor{currentstroke}{rgb}{0.000000,0.000000,1.000000}%
\pgfsetstrokecolor{currentstroke}%
\pgfsetdash{}{0pt}%
\pgfpathmoveto{\pgfqpoint{0.950000in}{3.391667in}}%
\pgfpathlineto{\pgfqpoint{1.183333in}{3.391667in}}%
\pgfusepath{stroke}%
\end{pgfscope}%
\begin{pgfscope}%
\pgftext[x=1.366667in,y=3.333333in,left,base]{\rmfamily\fontsize{12.000000}{14.400000}\selectfont Netlix, repr. users}%
\end{pgfscope}%
\begin{pgfscope}%
\pgfsetrectcap%
\pgfsetroundjoin%
\pgfsetlinewidth{1.003750pt}%
\definecolor{currentstroke}{rgb}{0.000000,0.500000,0.000000}%
\pgfsetstrokecolor{currentstroke}%
\pgfsetdash{}{0pt}%
\pgfpathmoveto{\pgfqpoint{0.950000in}{3.159333in}}%
\pgfpathlineto{\pgfqpoint{1.183333in}{3.159333in}}%
\pgfusepath{stroke}%
\end{pgfscope}%
\begin{pgfscope}%
\pgftext[x=1.366667in,y=3.101000in,left,base]{\rmfamily\fontsize{12.000000}{14.400000}\selectfont Movielens, repr. users}%
\end{pgfscope}%
\begin{pgfscope}%
\pgfsetrectcap%
\pgfsetroundjoin%
\pgfsetlinewidth{1.003750pt}%
\definecolor{currentstroke}{rgb}{1.000000,0.000000,0.000000}%
\pgfsetstrokecolor{currentstroke}%
\pgfsetdash{}{0pt}%
\pgfpathmoveto{\pgfqpoint{0.950000in}{2.926334in}}%
\pgfpathlineto{\pgfqpoint{1.183333in}{2.926334in}}%
\pgfusepath{stroke}%
\end{pgfscope}%
\begin{pgfscope}%
\pgftext[x=1.366667in,y=2.868000in,left,base]{\rmfamily\fontsize{12.000000}{14.400000}\selectfont Netflix, repr. items}%
\end{pgfscope}%
\begin{pgfscope}%
\pgfsetrectcap%
\pgfsetroundjoin%
\pgfsetlinewidth{1.003750pt}%
\definecolor{currentstroke}{rgb}{0.000000,0.750000,0.750000}%
\pgfsetstrokecolor{currentstroke}%
\pgfsetdash{}{0pt}%
\pgfpathmoveto{\pgfqpoint{0.950000in}{2.694000in}}%
\pgfpathlineto{\pgfqpoint{1.183333in}{2.694000in}}%
\pgfusepath{stroke}%
\end{pgfscope}%
\begin{pgfscope}%
\pgftext[x=1.366667in,y=2.635667in,left,base]{\rmfamily\fontsize{12.000000}{14.400000}\selectfont Movielens, repr. items}%
\end{pgfscope}%
\begin{pgfscope}%
\pgfsetbuttcap%
\pgfsetroundjoin%
\pgfsetlinewidth{1.003750pt}%
\definecolor{currentstroke}{rgb}{0.750000,0.000000,0.750000}%
\pgfsetstrokecolor{currentstroke}%
\pgfsetdash{{1.000000pt}{3.000000pt}}{0.000000pt}%
\pgfpathmoveto{\pgfqpoint{0.950000in}{2.461001in}}%
\pgfpathlineto{\pgfqpoint{1.183333in}{2.461001in}}%
\pgfusepath{stroke}%
\end{pgfscope}%
\begin{pgfscope}%
\pgftext[x=1.366667in,y=2.402667in,left,base]{\rmfamily\fontsize{12.000000}{14.400000}\selectfont Square Maxvol}%
\end{pgfscope}%
\end{pgfpicture}%
\makeatother%
\endgroup%